\documentclass[sigconf]{acmart}
\pdfoutput=1

\AtBeginDocument{%
  \providecommand\BibTeX{{%
    \normalfont B\kern-0.5em{\scshape i\kern-0.25em b}\kern-0.8em\TeX}}}

\usepackage{times}
\usepackage{soul}
\usepackage{url}
\expandafter\def\expandafter\UrlBreaks\expandafter{\UrlBreaks
  \do\a\do\b\do\c\do\d\do\e\do\f\do\g\do\h\do\i\do\j%
  \do\k\do\l\do\m\do\n\do\o\do\p\do\q\do\r\do\s\do\t%
  \do\u\do\v\do\w\do\x\do\y\do\z\do\A\do\B\do\C\do\D%
  \do\E\do\F\do\G\do\H\do\I\do\J\do\K\do\L\do\M\do\N%
  \do\O\do\P\do\Q\do\R\do\S\do\T\do\U\do\V\do\W\do\X%
  \do\Y\do\Z}

\usepackage{booktabs}

\usepackage[ruled]{algorithm2e} 

\SetAlFnt{\small}
\SetAlCapFnt{\small}
\SetAlCapNameFnt{\small}
\SetAlCapHSkip{1pt}
\IncMargin{-\parindent}

\usepackage{xcolor}

\usepackage{bbold}
\usepackage{amsmath,amssymb,amsfonts}
\usepackage{amsthm}

\newcommand{\prob}{\mathrm{Pr}}

\newcommand{\expected}[2]{\underset{#2}{\mathbb{E}}\Large[#1\Large]}

\usepackage{graphicx}

\usepackage[utf8]{inputenc}
\usepackage[english]{babel}
\newtheorem{lemma}{Lemma}
\newtheorem{theorem}{Theorem}

\newcommand{\femalevalue}{Equal price - Female valuable}
\newcommand{\femaleprice}{Expensive female - Equal value}

\usepackage{amsthm}
\newtheorem{definition}{Definition}

\newcommand{\E}{\mathbb{E}}
\newcommand{\msf}{\mathsf}

\usepackage{subcaption}

\newcommand{\m}{n_{\msf{m}}}
\newcommand{\w}{n_\msf{w}}

\graphicspath{{./images/}}

\renewcommand{\citeyearpar}{\cite}

\begin{document}
\title{Bidding Strategies with Gender Nondiscrimination Constraints for Online Ad Auctions}


\author{Milad Nasr}
\affiliation{%
  \institution{University of Massachusetts Amherst}
}
\email{milad@cs.umass.edu}
\author{Michael Carl Tschantz}
\affiliation{%
  \institution{International Computer Science Institute}
}
\email{mct@icsi.berkeley.edu}

\begin{abstract}

  Interactions between bids to show ads online can lead to an advertiser's ad being shown to more men than women even when the advertiser does not target towards men.  We design bidding strategies that advertisers can use to avoid such emergent discrimination without having to modify the auction mechanism.  We mathematically analyze the strategies to determine the additional cost to the advertiser for avoiding discrimination, proving our strategies to be optimal in some settings.  We use simulations to understand other settings.
  
  \end{abstract}

\maketitle

\section{Introduction}

Prior work found Google showing an ad for the Barrett Group,
a career coaching service promoting the seeking of high paying jobs,
more often to simulated men than women~\cite{datta15pets}.
Later work enumerates possible causes of this
disparity~\cite{datta18fatstar}.

One possibility, raised by Google itself~\cite{todd15pgh},
is that the Barrett Group targeted both men and women
equally, but other advertisers, on average, focused more on women, 
which would be in line with subsequent findings~\cite{lambrecht18ssrn}.
In this possibility, the Barrett Group found itself outbid for just
women by the
other advertisers who were willing to pay more than it was for reaching
women but not for men.
These other advertisers might be promoting products that many
find acceptable to target toward women, such as makeup.
Thus, it's possible that each advertiser's
targeting appears reasonable in isolation but interacts to bring about
emergent discrimination for a job-related ad.

For conscientious advertisers of products that should be broadcasted to
women and men at equal rates, such an outcome is unacceptable but
currently difficult to avoid.
While Google
offers the ability to skew ads toward men or toward women, it provides no
way to ensure that both men and women see the ad an equal number of times.
As discussed above, simply not targeting by gender is not enough
to guarantee parity.
Even running two ad campaigns of equal size is insufficient since the size
is determined by budget and not the number of ads shown, which means
that parity would only be achieved if women and men are equally
expensive to reach.

In this work, we consider how advertisers can ensure approximate demographic
parity for its ads without changing Google's ad auction mechanism,
which is based on a second-price auction~\cite{google18adsense}.
Given that an advertiser wishes to maximize its utility by reaching the
people most likely to respond to its ads, we model the advertiser's
utility function along with the parity goal as a constrained bidding
problem.
We consider both a very strict absolute parity constraint and a more
relaxed relative constraint inspired by the US EEOC's four-fifths rule
on disparate impact~\citep{eeoc1978}.
While using a second-price auction suggests that the advertisers
should bid their true value of showing an ad, a parity constraint
and multiple rounds of the auction interact to make deviations from
this truthful strategy optimal.
Intuitively, as in multi-round second-price auctions with budget
constraints~\cite{gummadi2011optimal}, it is sometimes better to bid less to
preserve the ability to participate in later auctions with a
lower cost of winning.
More interestingly, unlike with just budget constraints, it is also
sometimes better to bid more to ensure
an acceptable degree of parity, enabling participation in other
auctions later.

Given these complexities, finding an optimal bidding strategy for such
a constrained bidding problem is non-trivial.
We do so by modeling them as a Markov Decision Problems (MDPs).
Solving these MDPs using traditional methods, such as value iteration,
is made difficult by the continuous space of possible bid values over
which to optimize.
To avoid this issue, we find recursive formulae for each type of
constraint providing the optimal bid value and solve for their values
instead.
This approach allows us to solve the MDPs without needing to explicitly
maximize over the possible actions as in value iteration.

We compare this optimal constrained bidding strategy to the optimal
unconstrained strategy for both real and simulated data sets.
The cost to the advertiser for ensuring parity varies by setting, but
is manageable under the more realistic settings
explored.
In all cases, the revenue of the simulated Google remains roughly the
same or goes up.

By not modifying the core auction algorithm used by
Google and instead suggesting bidding strategies that could be deployed by
the advertisers, we believe this work provides a practical path towards
nondiscriminatory advertising.

\section{Related Work}


The most closely related work, recently 
looked at enforcing parity constraints with auction mechanisms, whereas
we do so with bidding strategies~\cite{celis19arxiv}.
While both approaches have their use cases, we believe ours is easier
to deploy since just the advertisers wanting the feature need to make
changes to implement it.
We further discuss tradeoffs between deployment approaches in
Section~\ref{sec:discussion}.
Our approach also differs by using strict constraints whereas theirs
uses probabilistic constraints.
Probabilistic constraints allow more utility but may be insufficient
in cases where approximate parity is required, as when disparate
impact is prohibited.
At an algorithmic level, they differ by using gradient decent.

A similar alternative approach could use auction mechanisms with Guaranteed ad
Delivery (GD)~\cite{salomatin2012unified,turner2012planning}.
An advertiser can act as two parties to the auction, one for each
gender, and use GD to ensure an equal number of wins for each party.
Unlike our bidding strategy, which an advertiser can
unilaterally employ, this approach requires
the ad exchange to change its auction mechanisms.

Prior works have looked at how to enforce (proportional) parity
constraints on the classifications produced by ML
algorithms~\cite{calders09icdmw,calders10dmkd,zemel13icml,kamishima12euromlkdd}.
We instead look at auctions.

Prior works have used MDPs to model ad slot auctions.
Li~et~al.~\citeyearpar{li2010value} and
Iyer~et~al.~\citeyearpar{iyer2011mean} have used them to find optimal
bidding strategies when advertisers do not know the exact values of
each type of ad slot and learn values by winning them.
They showed advertisers should overbid to learn more information.
Gummadi~et~al.~\citeyearpar{gummadi2011optimal} described the optimal
bidding strategy for the second-price auction in which each advertiser
has a limited budget, which leads to underbidding.
Zhang~et~al.~\citeyearpar{zhang2014optimal} derived optimal real-time
bidding strategies when each ad slot have different properties.

\section{Online Ad Auctions}

When a person visits a webpage, the webpage will often contain
dynamically loaded ads at fixed locations on the page.
These ads each occupy an ad \emph{slot}, a location at a time (or page
load) on the webpage.
In some cases, the website selects which ads to show in which slots
itself, such as with Facebook.
In other cases, the website contracts with a third-party,
to fill and charge for the slots in exchange for payments to the
website.
In either case, 
We call the entity choosing how to fill the slots an
\emph{ad exchange}.
For example, Google runs an ad exchange, Google Ad Manager, which
includes slots put up for sale by websites with its AdSense tool.

%

Typically, an ad exchange auctions off the slots it controls to
advertisers.
It can use \emph{real-time bidding} to auction off the slots
as the webpage loads.
The website and the ad exchange can offer advertisers various amounts
of information about the slot, such as the webpage it is on and
demographics about who is loading the page.
Advertisers performing \emph{programmatic advertising} use a dynamic
bidding strategy that adjusts their bids according to how well they
expect their ads to perform in the offered slot.
To avoid having to create programs for executing such strategies on
their own, advertisers often use a \emph{demand-side platform} (DSP).
Figure~\ref{fig:auction_scenario} demonstrates a sketch of the interactions. 
\begin{figure}
\centering
\includegraphics[width=\columnwidth]{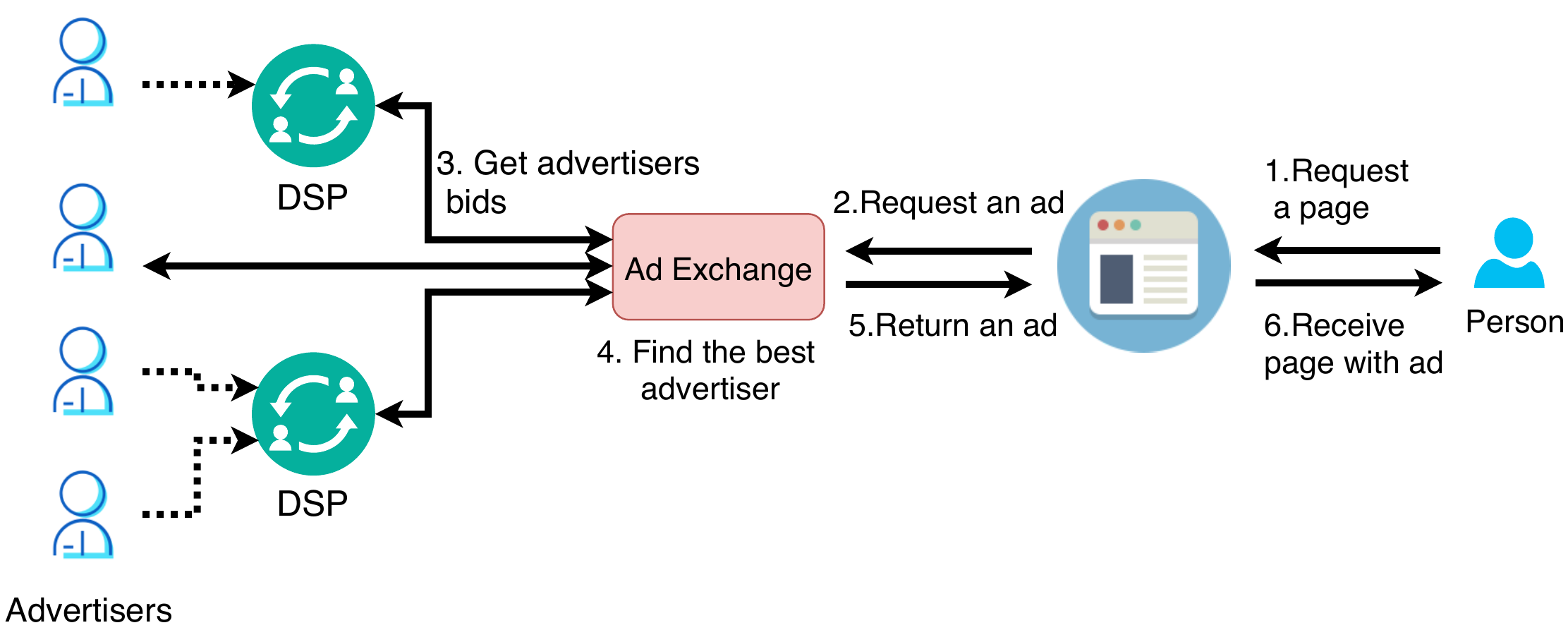}
\caption{Online advertisement interactions}
\label{fig:auction_scenario}
\end{figure}

An ad exchange may accept bids that are more complex than just a single
price, such as including an offer to pay a bonus if the website
visitor clicks the ad~\cite{google18adsense}.
Exchanges wishing to maximize the amount of bonuses it receives, or
to avoid annoying visitors, might consider the quality of the ad and
it's fit for the slot.
For simplicity, we will not consider these complications and instead
presume that all bids are simply offers to pay for showing the ad.

Second price auctions is a common mechanism for resolving such
auctions, with Google using a variation of one~\cite{google18adsense},
and we will presume the ad exchange uses one.
In this auction mechanism, the exchange selects the highest bidder as
the winner but only charges the bidder the price offered by the second
highest bidder.
Under certain circumstances,
this mechanism ensures that each bidder's optimal strategy is to bid
the actual amount it values the slot at, making the mechanism
\emph{truthful}.
Since ad exchanges sometimes sell more than one slot at time, such as
for a webpage with multiple slots, they often use generalized
second-price auctions, known as position
auctions~\cite{varian07ijio,edelman07aer}.

We model the above economy as a sequential game of incomplete information, where in each round of the game a set of self-interested rational advertisers bids to win an ad
slot through a second-price auction. 
We allow bids to vary over auctions and assume that each advertiser has a geometric lifespan.
For simplicity, we make the total number of advertisers $\alpha$ equal in all auctions by assuming that every time an advertiser dies a new advertiser joins.

At time $t$, each advertiser $i$ submits a bid $b_i^t$.
Let $b_{-i}^t$ be the bids of other the advertisers.
The ad exchange platform runs a second-price auction where
Advertiser $i$ wins the ad slot if its bid is higher than all other bids: $b_i^t > \max b_{-i}^t$.
For simplicity, we assume no ties, ensuring that such a winner exists.
Let $a_i^t$ be $1$ if the advertiser $i$ wins at round $t$ and be $0$ otherwise.
If the advertiser $i$ wins it will pay the second highest bid $d_i^t = \max b_{-i}^t$.
The cost of the auction $t$ is $c_i^t = a_i^t * d_i^t$ since the
advertiser $i$ only pays if it wins.

The ad slot auctioned at $t$ has a value $v_i^t$ for the advertiser $i$.
When an advertiser $i$ wins auction $t$, it gets an immediate reward, which is the value $v_i^t$ less its price $d^t_i$.
Thus, the utility of advertiser $i$ gained at each round is
$u_i^t = a_i^t*v_i^t - c_i^t = a_i^t(v_i^t - d_i^t)$.
Let the geometric parameter for the lifespan distribution for advertiser $i$ be $\delta_i$. 
The total utility for each advertiser is
$U_i = \sum_{t=0}^\infty \delta_i^{t} * a_i^t (v_i^t - d_i^t)$ 
where $\delta_i^{t}$ is exponentiation, not indexing like the others.



The advertiser $i$ should select its bids $b_i^t$ to maximize the expected value of $U_i$ where the expectation is over its value $v_i^t$
and the bids of other advertisers $b_{-i}^t$.
The advertiser can use market research, its prior experiences, and any
information provided by the ad exchange to estimate these uncertain
values.
In the case of a pure second-price auction, the values of
the other bids $b_{-i}^t$ are irrelevant and the optimal strategy is
to always set its bid $b_i^t$ equal to its estimation of its value
$v_i^t$.

However, this result does not carry over to all second-price auctions
with constraints, including the parity constraints we consider.
In this case, the behavior the other advertisers matters, but
estimating it for individual ad slots is difficult.
Furthermore, the advertiser is unlikely to estimate the value of every
ad slot individually even for a pure second-price auction.
Rather, the advertiser will likely model ad slots as each having a
\emph{type} belonging to a set $\Theta$ of reasonable size.
The types will represent the most important information to the
advertiser about the slot.
For simplicity, we will typically assume that $\Theta$ is equal to $\Gamma$.

For each type $\theta$, the advertiser will estimate the expected
value $v_i^\theta$ of a slot of type $\theta$.
For estimating the other bids, prior research~\cite{iyer2011mean} has
shown it reasonable to model them as coming a stationary fixed
distribution due to the large number of other advertisers.
To simplify the future analysis, we denote the CDF of other bids for a
slot of type $\theta$ by $g_i^\theta$.
Finally, let $p_i^\theta$ be the probability that the advertiser assigns to
type $\theta$.

With these estimations, we compute estimations of other key
quantities.
The probability of winning on auction $t$ for a slot of type $\theta$ with a bid of $x$ is
$q(x; g_i^{\theta}) = \prob (b_{-i}^t \leq x) = g_i^{\theta}(x)^{(\alpha-1)}$
where $\alpha$ is the number of advertisers at each ad slot auction.

The expected value of the utility for the advertiser $i$ for a single auction given the distribution of the other advertisers' bid $g_i^{\msf{w}}$ and $g_i^{\msf{m}}$ is
$\expected{u_i^t}{} 
= \sum_\theta p_i^\theta \times q(v_i^t; g_i^{\theta}) \times (v_i^t- d_i^t)$.
The expect value of the total utility for each advertiser is
\begin{align}
\expected{U_i}{}
&= \sum_t \delta_i^{t-1} \sum_\theta p_i^\theta  q(v_i^t; g_i^{\theta}) ( v_i^t- d_i^t)
\end{align}

\section{Parity Constraints}

Advertisers may have concerns in addition to attempting to maximize
the utility $U_i$, such as complying with laws and social norms.
In some cases, this will include ensuring that its ads reach various
protected groups to the same degree.
For example, an employer may desire that a job ad be shown to an equal
number of women and men to comply with laws prohibiting gender
discrimination in hiring~\cite{datta18fatstar}.
Such advertisers would like to place their bids in a manner to ensure
such demographic parity.

However, 
the above auction mechanism, as well Google's actual
mechanism as far as we can tell, does not offer any way of ensuring
that a job ad is shown to an approximately equal number of women and men, 
as required by laws prohibiting gender
discrimination in hiring~\cite{datta18fatstar}.
Furthermore, ad exchanges may be unwilling to support such
constraints given that only some advertisers have such concerns.
Thus, our goal is to provide advertisers with a bidding strategy that
dynamically adjusts bids to preserve the gender parity of the viewers,
which advertisers can unilaterally use without needing changes
to the auction mechanism of the ad exchange.

As an additional benefit of not modifying the ad auction mechanism,
our bidding strategy can be used for any type of auction.
However, we design and analyze them with only with second-price
auctions in mind.

To state our goal more precisely, we have to distinguish between
absolute (additive) and ratio (relative) parity.
An advertiser has \emph{$K$-strict absolute parity}, or
\emph{$K$-parity} for short, if after
each auction, the maximum difference between the number of auctions
that it wins for each gender is not more than $K$.
An advertiser has $R$-ratio parity, after each auction, if the maximum ratio of the number of auctions that it wins for each gender is not more than $R$.

Our goal is to find the optimal bidding strategy for advertisers
obeying either type of constraint.
This task is difficult since a constrained advertiser must consider not just
the immediate reward of winning a slot, but also how it may close or
open the possibility of winning additional slots later.
To see this, we will consider three examples involving a simplified
setting in which an advertiser $i$ is subject to $1$-parity and
knows exactly how long it will live.
In each example, it values men and women both at
$20$ (no variance), but that other advertisers value women at
an expected value of $21$ and men at an expected value of $5$.
This setting reflects that advertisers are willing to pay more, on
average, for women than men~\cite{lambrecht18ssrn}.

In the first example we consider, the advertiser knows that it will
live for exactly one ad auction.
In this case, the advertiser $i$ will bid the value of the immediate
reward $20$ that it receives for winning an auction regardless of
whether it is subject to a $1$-parity constraint since winning
the auction has no effect other than that immediate reward.
It will win an auction for a man and lose an auction for a woman.

Next, consider the advertiser's behavior for a series of two auctions.
The interesting case is two men in a row.
In this case, advertiser $i$ can only win one of the slots since it is
subject to a $1$-parity policy.
Thus, the utility of the advertiser will be smaller from having
$1$-parity, but it need not be half that of when it is unrestricted.
If the number of women is small enough ($p \ll 0.5$), the advertiser 
can assume it will get two men in a row and can lower the value of its bid
on the first man in hopes of winning at a discount, given the fluctuations
in the other advertisers' bids.
We call this \emph{underbidding}, although we emphasize that it is
underbidding with respect to its immediate reward, not with respect to
what is overall rational.
Underbidding effectively allows the advertiser $i$ to skip the first
auction if the variance in the other advertisers' bids produces an abnormally high
competing bid.
This is similar to how underbidding is optimal in some repeated
second-price auctions with a constrained budget~\cite{gummadi2011optimal}.
The degree of underbidding must balance the chance at getting a male
slot at a discount
with the risk of either losing both auctions
or getting a female slot for the second auction.

The opposite, \emph{overbidding}, can also occur.
To see this, consider a series of three auctions with a woman followed
by two men.
In this case, the advertiser $i$ can win both men, despite the
$1$-parity constraint, provided that it first wins the woman.
Thus, winning the woman produces not just an immediate reward, but also
a future reward by unlocking the ability to win more men.
If we presume negligible variance in the other bids, the advertiser $i$
will have to bid $22$ to win the woman and pay the second price of $21$, yielding an immediate reward of $-1$ by
bidding $1$ over the inherent value $20$ of the female ad slot to the advertiser.
However, since the immediate reward of a male slot is $20-5 = 15$, being
able to win the second man means a net positive gain of $15 - 1$.
(We ignore the effects of underbidding since we are now considering
negligible variance in the other advertiser's bids, which makes the
effect go away.)

We find this distinction between the immediate reward and the
future rewards coming from future flexibility useful for determining
the optimal bidding strategy.
However, doing so requires not only making the above intuitions
quantitative, but also dealing with additional probabilistic factors,
such as the genders of ad slots not being known in advance and the
uncertain duration of the auction sequence.
To overcome these difficulties, we switch to a more systematic model
for each type of constraint.

\section{Absolute Parity Constraints}

An advertiser want to show an ad to equal numbers of men and women.
A particularly careful advertiser may desire that this parity
constraint holds not only at the end of ad campaign but throughout.
Such continuous parity ensures that the advertiser would pass an audit
checking for this property at any point in time.
It also ensures meeting the parity goal if the the ad campaign must be
cut short or if a sudden influx of competing advertisers prevents
winning addition slots.

Meeting this strict goal is impossible since the first ad must go to
either a man or woman, and not both.
To account for this, we relax this goal by allowing a difference to
arise.
We use $K$ to denote the maximum allowed difference where $K=1$ is the
strictest constant compatible with showing any ads.


To make this precise, 
we let $\Gamma$ denote a set of groups.
We are typically interested in the case where
$\Gamma = \{\msf{m},\msf{w}\}$ with $\msf{m}$ denoting men and
$\msf{w}$ women.
In this case, we use $p$ to denote the probability of a male ad slot
(i.e., $p_i^{\msf{m}}$).
We use $n_i$ to denote the number ad slots for people in group $i$ won
by the constrained advertiser.

\begin{definition}[$K$-parity]
\label{def:strict}
An advertiser obeys a \emph{$K$-absolute parity} constraint, 
or \emph{$K$-parity} for short,
for a set of groups $\Gamma$ 
iff, after each auction, for all groups $i$ and $j$ in $\Gamma$,
the number of auctions that it wins satisfies 
$n_{i} - n_{j} \leq K$.
\end{definition}

We study approximating the optimal bidding stagy that an advertiser
desiring to meet a $K$-parity constraint can use to do so.
In our analysis, we assume all of the advertisers have an unlimited
budget.
Thus, they can bid on all auctions in its lifespan, unless maintaining
$K$-parity constraint precludes it.


\subsection{Modeling} \label{sec:model}

To find the optimal bidding strategy for the $K$-parity advertiser, we model the problem as a Markov Decision Problem (MDP).
The obvious state space for such an MDP would have states of the form 
$\langle \m,\w,\theta \rangle$,
where $\m$ and $\w$ is the current number male and female viewers,
respectively, and $\theta$ is the type of the ad slot currently being
auctioned off, which we presume corresponds to a gender.
($\theta$ could be generalized to allow targeting toward certain men
and women.)
Observing that only $\m - \w$ matters,
we instead use a smaller space of $|\Theta|\times(2K+1)$ states. 
We denote each state by a tuple $\langle k, \theta\rangle$, where $k$ is the difference between male and female viewers.
When the advertiser wins an ad slot for a male viewer, the advertiser goes from state $k$ to $k+1$; for a female, it goes from $k$ to $k-1$.
The value of the $\theta$ is decided by a random process depending upon the value of $p$, where $p$ is probability of the viewer being male.

To find the optimal solution, we write the Bellman equation for the MDP in the steady state. Since we consider the steady state regime we also replace the value of each ad slot by its expected value (i.e., $v_i^\theta$).
The value function for each state except for two states $\langle K, \msf{m}\rangle$ and $\langle -K, \msf{w}\rangle$ has two parts:
a reward function $R$ that indicates the immediate reward of taking action $b_i$ and 
$N$ that is the future value the advertiser gets by doing that action.
We write the value functions as follows:
\begin{align}
V(k, \theta; g_i) =
 \max_{b_i} \Big\{ R^{\theta}(b_i; g_i^{\theta}) + \delta N^{\theta}(b_i,k; g_i) \Big\} \label{eqn:value}
\end{align}

\begin{align}
R^{\theta}(b_i; g_i^{\theta})&= q(b_i; g_i^{\theta}) (v^{\theta}_i - d_i^{\theta}) \label{eq:reward}
\end{align}

$N^{\theta}(b_i,k; g_i)$, the future value that advertiser $i$ gets by bidding $b_i$ at state $\langle k, \theta\rangle$, 
consists of two part with the first part $N^\theta_{\msf{win}}$ being the value that the advertiser gets if it wins and the second part $N^\theta_{\msf{lose}}$ being the value when it loses. 
We treat $g_i$ as providing both $g_i^{\msf{m}}$ and $g_i^{\msf{w}}$.
$R^{\msf{w}}(b_i; g_i^{\msf{w}})$ and $R^{\msf{m}}(b_i; g_i^{\msf{m}})$ show the reward value that advertiser $i$ will receive if it wins an ad slot auction viewed by female and male. We have:
\begin{align*}
N^{\theta}(b_i,k; g_i) = 
q(b_i; g_i^{\theta}) * N^\theta_{\msf{win}}(k; g_i)
+ (1-q(b_i; g_i^{\theta})) * N^\theta_{\msf{lose}}(k; g_i)
\end{align*}
with
\begin{align*}
N^{\msf{m}}_{\msf{win}}(k; g_i) 
&= p V(k+1, \msf{m}; g_i) + (1-p) V(k+1, \msf{w}; g_i)\\
N^{\msf{w}}_{\msf{win}}(k; g_i) 
&= p V(k-1, \msf{m}; g_i) + (1-p) V(k-1, \msf{w}; g_i)\\
N^\theta_{\msf{lose}}(k; g_i)
&= p V(k, \msf{m}; g_i) + (1-p) V(k, \msf{w}; g_i)
\end{align*}
As for the two edge cases, their values are solely determined by the values of their successor states since the advertiser cannot win the current auction:
\begin{align*}
V(K,\msf{m}; g_i) &=
\delta * \left( pV(K,\msf{m}; g_i) + (1-p) V(K,\msf{w};g_i) \right) \\
V(-K,\msf{w}; g_i) &=
\delta * \left( pV(-K,\msf{m};g_i) + (1-p) V(-K,\msf{w}; g_i) \right)
\end{align*}

\subsection{Computing Optimal Bidding Strategies}
\label{sec:computing_optimal}

Computing $V$ with MDP solvers, such as value iteration, is
complicated by the bid space being continuous.
Computing $V$ for a discretization of this space will require a fine
discretization to avoid rounding errors, which will mean slow
convergence.
Using numerical optimization methods is complicated by $V$ not being a
linear function in $b_i$.
To avoid these complexities, we instead rewrite $V$ in a form that can
be solved without any optimization.

To identify the optimal bidding strategy,
we observe that the two edge cases do not involve a decision and
the strategy of bidding $0$ is forced for them.
We also observe that for the remaining states the valuation function
\eqref{eqn:value} includes many terms that do not change under various
bidding strategies. 
We collect these constants into a term $\Lambda_i$, which we can ignore
while optimizing the strategy.
We replace $q(b_i;g_i^\theta)d_i^t$ by $c(b_i;g_i^\theta)$ that indicates the estimated cost of each ad slot.
The remainder of the valuation function provides the 
\emph{conjoint valuation} function $\Phi_i^\theta$.
In more detail,
\begin{align}
\!\!\!V(k,\theta; g_i)
&= \max_{b_i} \Big\{ q(b_i;g_i^{\theta})\Phi_i^{\theta}(k;g_i) - c(b_i;g_i^{\theta}) + \Lambda_i(k;g_i) \Big\} \notag\\
&= \max_{b_i} \Big\{ q(b_i;g_i^{\theta})\Phi_i^{\theta}(k;g_i) - c(b_i;g_i^{\theta}) \Big\}+ \Lambda_i(k;g_i)\label{eq:v_theta}
\end{align}
where 
\begin{align*}
\Lambda_i(k;g_i^{\msf{m}},g_i^{\msf{w}}) =  \delta (pV(k,\msf{m}; g_i)
+ (1-p)V(k,\msf{w}; g_i) )
\end{align*}
The conjoint valuation $\Phi$ represents the reward for winning, both immediate and long-term, which is why it is multiplied by the probability of winning $q(b_i;g_i^{\theta})$.
The expected cost of winning $c(b_i; g_i^{\theta})$ is subtracted from this product.
$\Phi$ breaks down along the lines of winning and losing cases, as $N$ did:
\begin{align}\label{eq:conj_theta}
\Phi^{\theta}(k; g_i) = v_i^{\theta} + \delta (\Phi_{\msf{win}}^{\theta}(k; g_i) - \Phi^{\theta}_{\msf{lose}}(k; g_i))
\end{align}
where
\begin{align*}
\Phi^{\msf{m}}_{\msf{win}}(k; g_i) &= pV(k+1,\msf{m}; g_i) + (1-p)V(k+1,\msf{w}; g_i)\\
\Phi^{\msf{w}}_{\msf{win}}(k; g_i) &= pV(k-1,\msf{m}; g_i) + (1-p)V(k-1,\msf{w}; g_i)\\
\Phi^{\theta}_{\msf{lose}}(k; g_i) &= pV(k, \msf{m}; g_i) + (1-p)V(k,\msf{w}; g_i)
\end{align*}
The term $v_i^\theta$ represents the immediate value of winning the ad slot.
The reminder considers the gain that the advertiser gets from the future by winning (moving to a new state) or losing (staying put).
The difference between future rewards for winning and those for losing corresponds to the amount of overbidding called for, which explains the subtraction in \eqref{eq:conj_theta}.

The following theorem shows the usefulness of this decomposition.  It uses the following lemma:
\begin{lemma}[\citeauthor{iyer2011mean} \citeyear{iyer2011mean}]\label{lemma:max_f}
For any continues non-decreasing function $q(x)$ on $[0,1]\times[0,1]$,  function $f(x,v)= q(x)(v-x)+ \int_{0}^x q(u) \mathop{}\!\mathrm{d}u$ gains its maximum when $x=v$.
\end{lemma}
\begin{theorem}\label{thm:optimal_kparity}
For any given $g_i$ and $K$, the optimal bid at all states $\langle k, \theta\rangle$ other than the edge cases $\langle K, \msf{m}\rangle$ and $\langle -K, \msf{w}\rangle$ is $\Phi^{\theta}_i(k; g_i)$.
\end{theorem}
\begin{proof}
Without loss of generality we assume all of the bids are between $0$ and $1$. 
The bidding strategy that maximize the equation~\eqref{eq:v_theta} will be the optimal strategy. 
To maximize this equation, we can omit the $\Lambda_i$ function since it is constant for each $b_i$. 
Similar to~\cite{iyer2011mean}, we rewrite the cost function $c(b_i;g_i^\theta)$ as 
\begin{align*}
  c(b_i;g_i^\theta) = q(b_i;g_i^\theta) b_i - \int_{0}^{b_i} q(u;g_i^\theta) du
\end{align*}

Now, we can rewrite the decision problem of the advertiser $i$ as
\begin{align*}
&\max_{b_i} \Big\{q(b_i;g_i^{\msf{m}})\Phi_i^{\msf{m}}(k;g_i^{\msf{m}},g_i^{\msf{w}}) - c(b_i;g_i^{\msf{m}})\Big\} \\
&=\max_{b_i} \Big\{ q(b_i;g_i^{\msf{m}})\Phi_i^{\msf{m}}(k;g_i) - \Big(   q(b_i;g_i^{\msf{m}})b_i - \int_0^{b_i}  q(u;g_i^{\msf{m}}) \mathop{}\!\mathrm{d}u \Big) \Big\}\\ 
&=\max_{b_i} \Big\{ q(b_i;g_i^{\msf{m}})(\Phi_i^{\msf{m}}(k;g_i)- b_i )+ \int_0^{v_i}  q(u;g_i^{\msf{m}}) \mathop{}\!\mathrm{d}u \Big\} 
\end{align*}

We know $q(b_i;g_i^{\msf{m}})$ is a continues non-decreasing function. 
Therefore, we can use Lemma~\ref{lemma:max_f} with $q(x) = q(x;g_i^{\msf{m}})$
to conclude equation~\eqref{eq:v_theta} is at its maximum when the bid is $\Phi_i^{\msf{m}}(k;g_i)$ for $\theta = \msf{m}$.
Similarly we can show for equation~\eqref{eq:v_theta} that the optimal bid is $\Phi_i^{\msf{w}}(k;g_i)$ where $\theta = \msf{w}$. 
\end{proof}

This theorem means that we do not need to search the space of possible bid values to find the optimal bid.
Rather, we can just compute the optimal bid using $\Phi$.
While $\Phi$ depends upon the value function $V$, we can recursively make use of this fact to compute $V$ without such a search either.
In particular, the theorem implies that 
\begin{align*}
V(k, \theta; g_i) &= R^{\theta}(\Phi_i^\theta(k; g_i); g_i^{\theta}) + \delta N^{\theta}(\Phi_i^\theta(k; g_i), k; g_i)
\end{align*}

However, this equation is still not a closed form solution.
Thus, Algorithm~\ref{alg:converge} does this calculation iteratively to converge to the states' values. 
Although, showing the convergence in general is an open problem, as discussed in Section~\ref{sec:experiments}, our experiments find convergence within a reasonable tolerance within a feasible number of iterations.
\begin{algorithm}[t]
\SetAlgoLined
\DontPrintSemicolon

\KwIn{$K,g_i, \alpha, v^m, v^w, \epsilon$}
    {\textbf{Initialize} $V[-K:K, \msf{m}] \gets \frac{v^m + v^w}{2}; V[-K:K, \msf{w}] \gets \frac{v^m + v^w}{2}$}

        {\Repeat{$\Delta < \epsilon$}{
        {$\Delta \gets 0$}

        {\For{$k$ \textup{in} $\{-K, \ldots, K\}$}
          {
            {\For{$\theta$ \textup{in} $\{\msf{m}, \msf{w}\}$}
              {
                {$V'[k, \theta] \gets R^\theta(\Phi^{\theta}(k); g_i^{\theta}) + \delta N^{\theta}(\Phi^{\theta}(k), k; g_i)$}
                {$\Delta \gets \max(\Delta,\: |V'[k,\theta] - V[k,\theta]|)$}
              }
            }
          }
        }
        {$V \gets V'$}
        }}

    \caption{Iterative approach to find $V$}
    \label{alg:converge}
\end{algorithm}

To use our approach, an advertiser (or DSP) runs
Algorithm~\ref{alg:converge} to compute the value function $V$ and
stores it as a look-up table.
Then, for each new ad auction, the advertiser first checks if it
winning the auction would violate the parity constraint.
If so, it will not participate in the auction (i.e., bids zero).
Otherwise, The advertiser bids the value of $\Theta_i^\theta(k)$,
which can be easily computed from value functions.

\section{Ratio Constraints}
\label{sec:ratio}

While constraints on the difference between the number of ads shown to
each gender are intuitive, the EEOC's four-fifths rule found in US
regulations against disparate impact in employment instead focuses on
a ratio~\citep{eeoc1978}.
The ratio considered is not simply between the number of ads shown to
each gender.
Rather, it acknowledges that parity can be unrealistic due to having
differing numbers of male and female applicants.
It adjusts for that factor by comparing the fraction of female
applicants receiving a job offer to the fraction male applicants
receiving a job offer.
It requires that this ratio of ratios be between $5/4$ and $4/5$.
Similarly, our ratio constraint compares two ratios, checking whether
the fraction of female ad slots won is within a factor of $r$ to the
fraction of male ad slots won.


Strictly enforcing this check creates problems when the number of
slots seen so far is small since the fractions won may be very far
apart even when the number of ads shown to each gender only differs by
$1$.
To avoid this issue, we also allow an additive difference in the
number of ads show to each gender.
The resulting rule may be viewed as a hybrid between a pure ratio
constraint and the absolute constraint we have already presented.


We use similar notation as in Section~\ref{sec:model} to express this
constraint in a manner that avoids division by zero.
\begin{definition}[$(r,K)$-ratio]
\label{def:rk}
An advertiser obeys a $(r,K)$-ratio constraint, 
for a set of groups $\Gamma$ 
iff, after each auction, for all groups $i$ and $j$ in $\Gamma$,
the number of auctions that it wins satisfies 
$r p_{i} n_{j} \leq p_{j} n_{i} + K$
where $p_{i}$ and $p_{j}$ is 
the probability of seeing slots for groups $i$ and $j$, respectively.
\end{definition}


\subsection{Modeling}

Similar to the $K$-parity constraint, we limit ourselves to the case
where $\Gamma$ and $\Theta$ only contain two types, which we treat as male and
female.
We use $p$ as the probability of a male.
We denote each state by a triplet $\langle \m,\w,\theta \rangle$,
where $\m$ and $\w$ is the current number male and female viewers,
respectively.

While we reuse the immediate reward function 
$R^{\theta}$ from~\eqref{eq:reward},
we rewrite the value function $V$ and future value function $N$.
When winning the slot would not violate the constraint,
\begin{align*}
V(\w,\m,\theta; g_i)
&= \max_{b_i} \left\{ R^\theta (b_i;g_i^\theta)  + \delta N^\theta(b_i,\m,\w;g_i) \right\}
\end{align*}
When offered a male that may not be won because $r (1-p) (\m+1) > p \w + K$ where $\m$ is the current number of males won,
\begin{align*}
V(\w,\m,\msf{m}; g_i) 
&= \delta \left( p V(\m,\w,\msf{m};g_i) + (1-p) V(\m,\w,\msf{w};g_i) \right)
\end{align*}
When a female may not be won since $r p (\w+1)  > (1-p) \m + K$,
\begin{align*}
V(\w,\m,\msf{w}; g_i) 
&= \delta \left(p V(\m,\w,\msf{m};g_i) + (1-p) V(\m,\w,\msf{w};g_i) \right)
\end{align*}
We call these two cases \emph{edge cases}.


We set the future value $N^{\theta}(b_i,\m,\w; g_i)$ at
\begin{align*}
q(b_i; g_i^{\theta}) * N^\theta_{\msf{win}}(\m,\w; g_i)
+ (1-q(b_i; g_i^{\theta})) * N^\theta_{\msf{lose}}(\m,\w; g_i)
\end{align*}
with
\begin{align*}
    N^{\msf{m}}_{\msf{win}}(\m,\w; g_i) 
    &= p V(\m{+}1,\w,\msf{m}; g_i) + (1{-}p) V(\m{+}1,\w, \msf{w}; g_i)\\
    N^{\msf{w}}_{\msf{win}}(\m,\m; g_i) 
    &= p V(\m,\w{+}1, \msf{m}; g_i) + (1{-}p) V(\m,\w{+}1, \msf{w}; g_i)\\
    N^\theta_{\msf{lose}}(\m,\w; g_i)
    &= p V(\m,\w, \msf{m}; g_i) + (1{-}p) V(\m,\w, \msf{w}; g_i)
\end{align*}

\subsection{Computing Optimal Bidding Strategies}

We use a similar approach as in Section~\ref{sec:computing_optimal} to
find optimal strategies.
As before, we force the strategy to bid zero when winning would
violate the constraint and do not include these cases in the
optimization.
We rewrite the value $V(\m,\w\theta; g_i)$ as 
\begin{align*}
\max_{b_i} \Big\{q(b_i;g_i^{\theta})\Phi_i^{\theta}(\m,\w;g_i) - c(b_i;g_i^{\theta}) \Big\}
  + \Lambda_i(\m,\w;g_i)
\end{align*}
where 
\begin{align*}
\Lambda_i(\m,\w; g_i^{\msf{m}},g_i^{\msf{w}}) 
&= \delta (pV(\m,\w,\msf{m}; g_i) + (1-p)V(\m,\w,\msf{w}; g_i) )
\end{align*}
and
\begin{align*}
\Phi^{\theta}(\m,\w; g_i) 
&= v_i^{\theta} + \delta (\Phi_{\msf{win}}^{\theta}(\m,\w; g_i) - \Phi^{\theta}_{\msf{lose}}(\m,\w; g_i))
\end{align*}
where
\begin{align*}
\Phi^{\msf{m}}_{\msf{win}}(\m,\w; g_i) &= p V(\m{+}1,\w ,\msf{m},; g_i) + (1{-}p) V(\m{+}1,\w, \msf{w}; g_i)\\
\Phi^{\msf{w}}_{\msf{win}}(\m,\w; g_i) &= p V(\m,\w{+}1 ,\msf{m}; g_i) + (1{-}p) V(\m,\w{+}1, \msf{w}; g_i)\\
\Phi^{\theta}_{\msf{lose}}(\m,\w; g_i) &= pV(\m,\w ,\msf{m}; g_i) + (1{-}p)V(\m,\w,\msf{w}; g_i)
\end{align*}

\begin{theorem}
\label{thm:optimal_ratio}
For all $r$, $K$, groups $i$,
and states $\langle \m,\w, \theta \rangle$ other than the edge cases,
the optimal bid is $\Phi^{\theta}_i(\m,\w; g_i)$.
\end{theorem}
The proof is similar to that of Theorem~\ref{thm:optimal_kparity}.

This theorem eliminates the need for searching the space of possible
bids at each state to find the optimal one.
Whereas we could bound the state space for $K$-parity by
tracking the difference $k$ instead of the actual numbers of male and
female ad slots won, we cannot similarly bound the state space for the
$(r,K)$-ratio constraint.
In practice, however, each advertiser either has a limited budget or
is advertising for a limited time allowing us to estimate a finite set
of reachable states.
We use $\mu$ to indicate estimated the maximum number of male ad slots won
in our experiments.
Algorithm~\ref{alg:converge_ratio} computes the value of each state
reachable assuming $\mu$.

\begin{algorithm}[t]
    \SetAlgoLined
    \DontPrintSemicolon
    
    \KwIn{$r,K,p,g_i, \alpha, v^m, v^w, \epsilon,\mu $}
        {\textbf{Initialize} $V[0 : M,0: r\frac{(1-p)}{p}M + K ,\msf{m}] \gets \frac{v^m + v^w}{2};$ \mbox{}\hspace{8em}\phantom{\textbf{Initialize}} $V[0 : M,0: r\frac{(1-p)}{p}M + K ,\msf{w}] \gets \frac{v^m + v^w}{2}$}
    
            {\Repeat{$\Delta < \epsilon$}{
            {$\Delta \gets 0$}
    
            {\For{$\m$ \textup{in} $\{0, \ldots, \mu \}$}
                {\For{$\w$ \textup{in} $\{0, \ldots , r\frac{(1-p)}{p}\mu + K \}$}
                {

                    {\For{$\theta$ \textup{in} $\{\msf{m}, \msf{w}\}$}
                    {
                        {$V'[\m,\w, \theta] \gets R^\theta(\Phi^{\theta}(\m,\w); g_i^{\theta}) + \delta N^{\theta}(\Phi^{\theta}(\m,\w), \m,\w; g_i)$}
                        {$\Delta \gets \max(\Delta,\: |V'[\m,\w,\theta] - V[\m,\w,\theta]|)$}
                    }
                    }
                }
                }
            
            }
            {$V \gets V'$}
            }}
    
        \caption{Iterative approach to find $V$}
        \label{alg:converge_ratio}
\end{algorithm}

An advertiser using our approach, does so in the same manner as with
our approach to parity constraints.
That is, it first runs Algorithm~\ref{alg:converge_ratio} and stores
$V$ as a look-up table.
It skips auctions when winning would violate the constraint and
otherwise bids $\Theta_i^\theta(k)$, computed from $V$.

We can extend this approach to recover if the 
advertiser underestimates $\mu$.
In this case, the advertiser can use a linear approximation to
estimate the optimal bid.
To do so, let $\rho = \frac{\w}{\m} (\mu-1)$.
If $\rho$ is an integer value, then the advertiser bids
$\Phi^{\theta}(\rho, \mu-1)$.
Otherwise, the advertiser bids
$\Phi^{\theta}(\lfloor\rho\rfloor,\mu-1) 
 \,+\, (\rho - \lfloor\rho\rfloor)
  * (\Phi^{\theta}(\lceil\rho\rceil,\m) - \Phi^{\theta}(\lfloor\rho\rfloor,\m))$.





\section{Experiments}
\label{sec:experiments}
We simulate various scenarios to show the
feasibility of our method and to measure the impact of our fairness constraints
on utility.
To do so, we implemented a second-price auction simulator in Python,
where each advertiser gets the gender of the website viewer
before selecting its bid and participating
in the ad slot auction.
To simulate the viewer, we draw their genders
independent and identically from a binomial distribution with
probability $p$ where $p$ is the probability of the viewer be male.
 
We focus on a single advertiser $i$ and measure how its utility
changes when it has either one of our fairness constraints or not.
When having fairness constraints, it uses our bidding strategy, with
$\delta$ set to $0.999$ (unless otherwise noted) and $\epsilon$ set to $0.001$.
When not, it bids it immediate value $v_i^t$ for the ad slot $t$, as is
rational for an unrestricted second-price auction.
We assume that the other advertisers are unrestricted, that they
always bid their values.
To obtain distributions over ad values, we used both a real dataset (The Yahoo!\@ A1 Search Marketing Advertiser Bidding Dataset)
and a simulated one. The Yahoo!\@ A1 data does not have exact timestamp so we could not use it to estimate the number of advertisers (i.e., $\alpha$) for each ad auction. To estimate $\alpha$, we visited top websites\footnote{based on \url{https://www.alexa.com/topsites}} that have ads using header bidding method~\cite{sayedi2018real} for one month (June 2019) and collected how many advertisers bid on a specific ad slot. In our experiments we never saw more than 10 advertisers bid on an ad slot suction. In line with our observation, we assume that there are $\alpha=10$
advertisers bidding for each ad slot.

\subsection{Real Dataset}
 
The Yahoo!\@ A1 Search Marketing Advertiser Bidding Dataset contains
anonymized bids of advertisers participating in Yahoo!\@ Search
Marketing auctions for the top 1000 search queries from June 15, 2002,
to June 14, 2003.
The dataset includes 18~millions bids from more than 10,000
advertisers, but without the exact timestamps or information about the
ad viewer.
Each record in this dataset indicates a course timestamp with 15
minutes precision, the advertiser, the keyword, and the bid.
\begin{figure}
\includegraphics[scale=0.4]{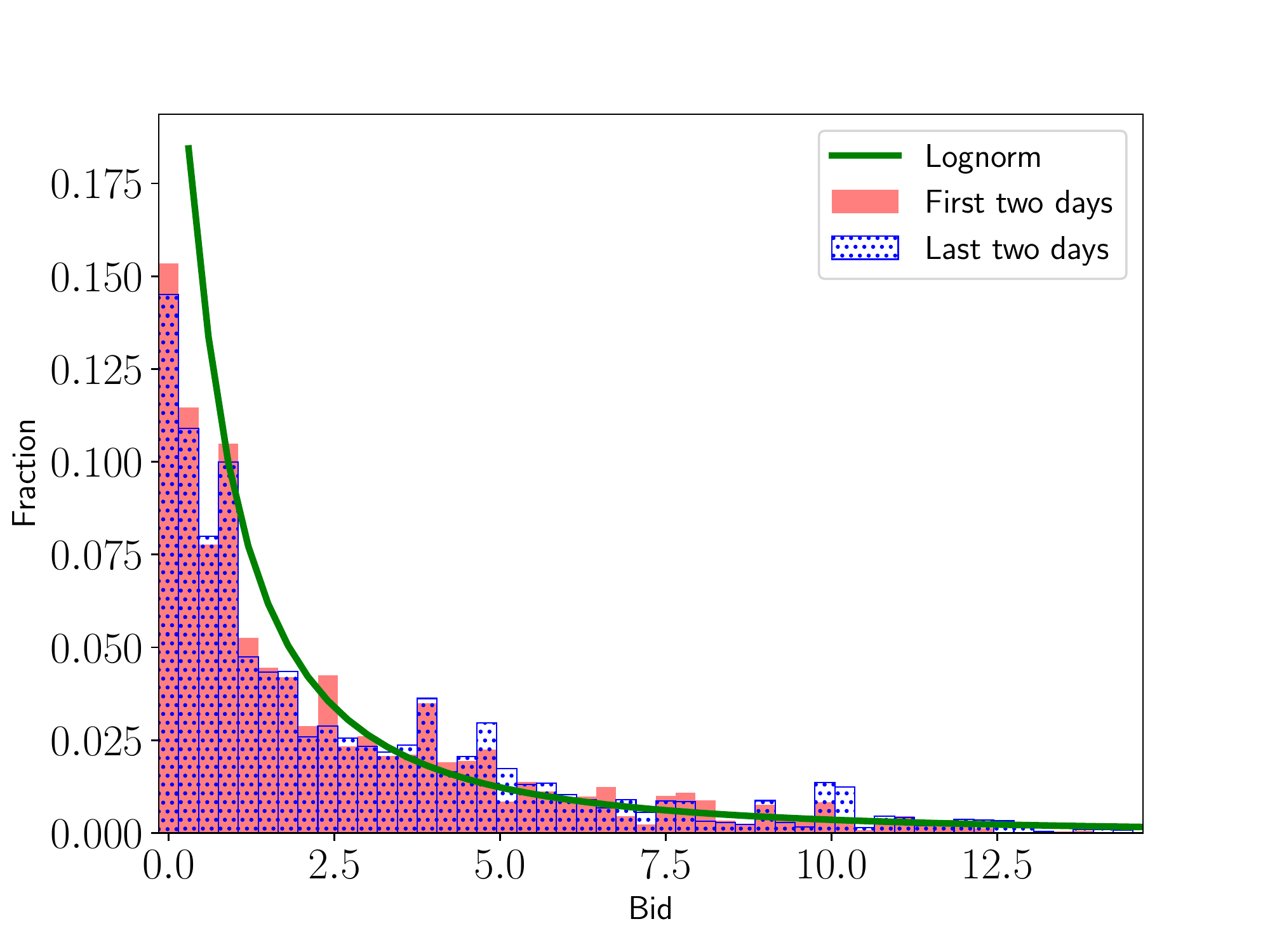}
\caption{Distribution of the bids for two different periods. The distributions are similar in both periods which supports the assumption that bids distribution are stationary.}
\label{fig:bid_distribution}
\end{figure}
 
In our analysis, we assumed bids have stationary distribution.
We evaluate this assumption on our dataset.
We use a specific keyword (keyword number 2) and we gathered all of
bids from different advertisers in four days period (starting
2/15/2003).
Then, we compute the empirical distribution of the bids of the first
two days and the second two days.
Figure~\ref{fig:bid_distribution} presents the distribution of the
bids for these periods, showing that the distributions are very
similar in both periods, supporting our stationarity assumption.
The figure also shows that the bids follow a log-normal distribution,
in line with the findings of Balseiro~et~al.~\citeyearpar{balseiro2017budget}.
 
\begin{figure*}
 \centering
 \hfill
\begin{subfigure}{.3\linewidth}
  \centering
  \includegraphics[scale=.28]{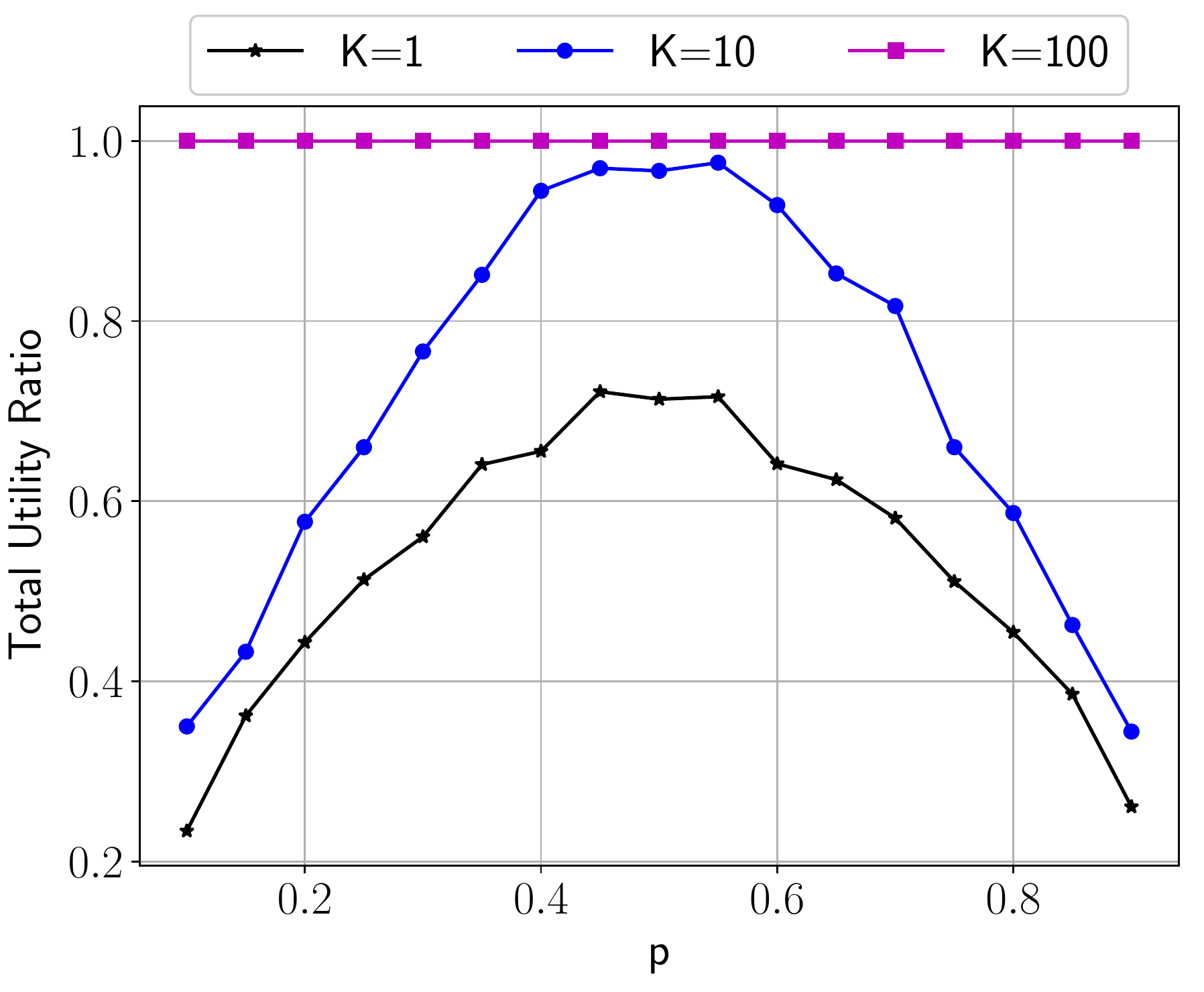}
  \caption{Cost of $K$-parity }
  \label{fig:kparity-yahoo}
\end{subfigure}
\hfill
\begin{subfigure}{.3\linewidth}
 \centering
  \includegraphics[scale=.28]{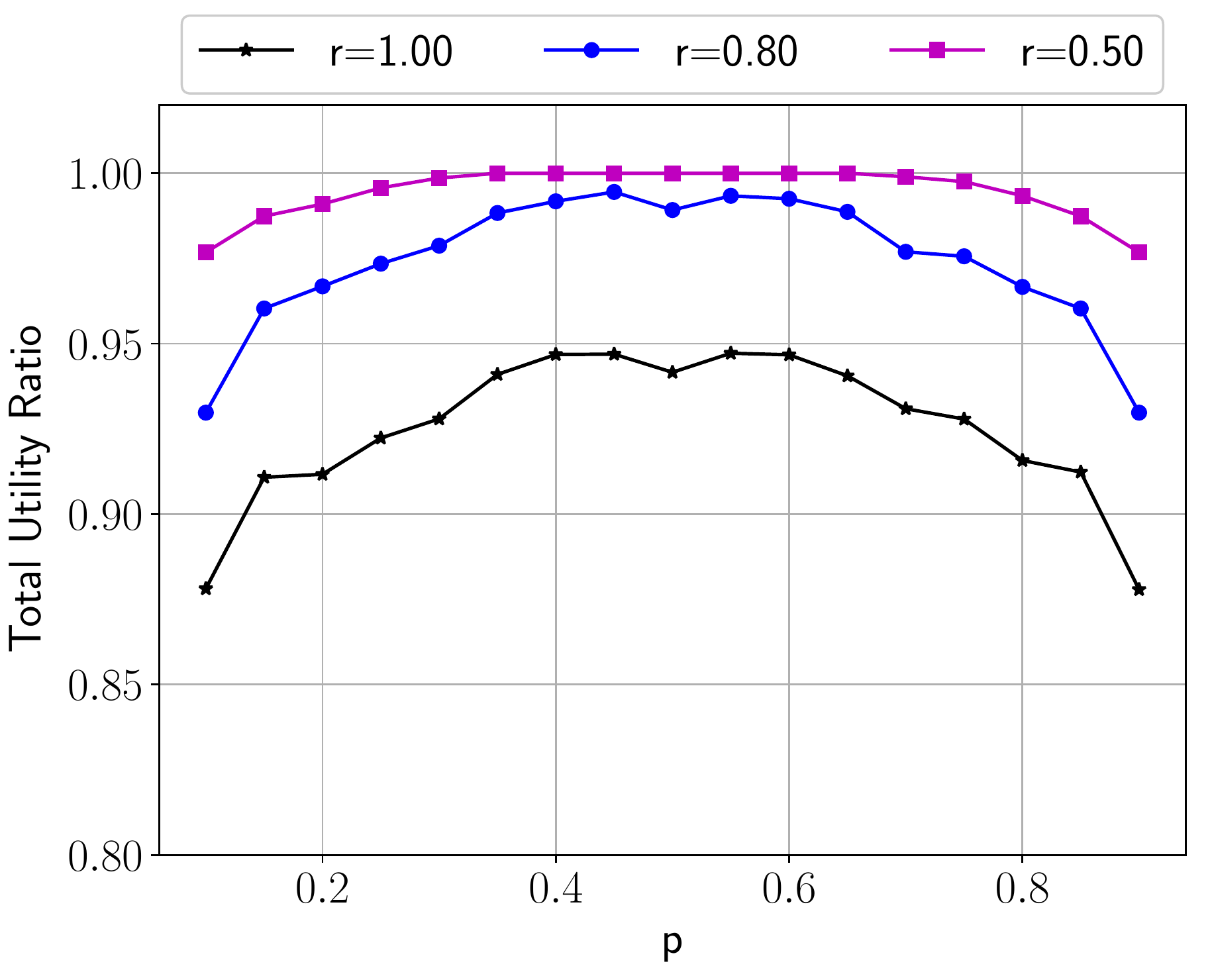}
  \caption{Cost of $(r,5)$-ratio}
  \label{fig:rratio-yahoo}
\end{subfigure}
\hfill
\begin{subfigure}{.3\linewidth}
 \centering
  \includegraphics[scale=.28]{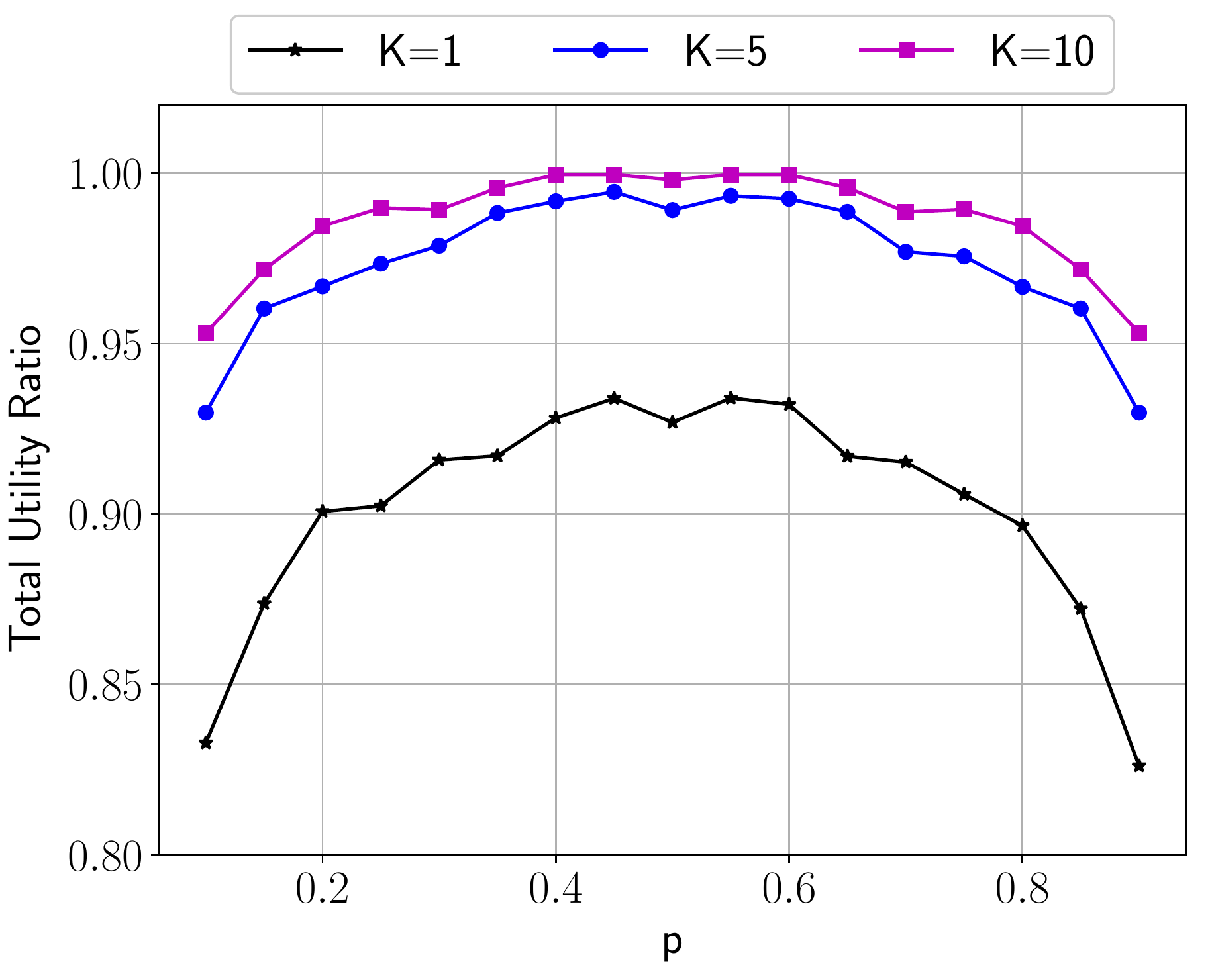}
  \caption{Cost of $(0.8,K)$-ratio}
  \label{fig:rratio-yahoo}
\end{subfigure}

\hfill
\begin{subfigure}{.45\linewidth}
 \centering
  \includegraphics[scale=.3]{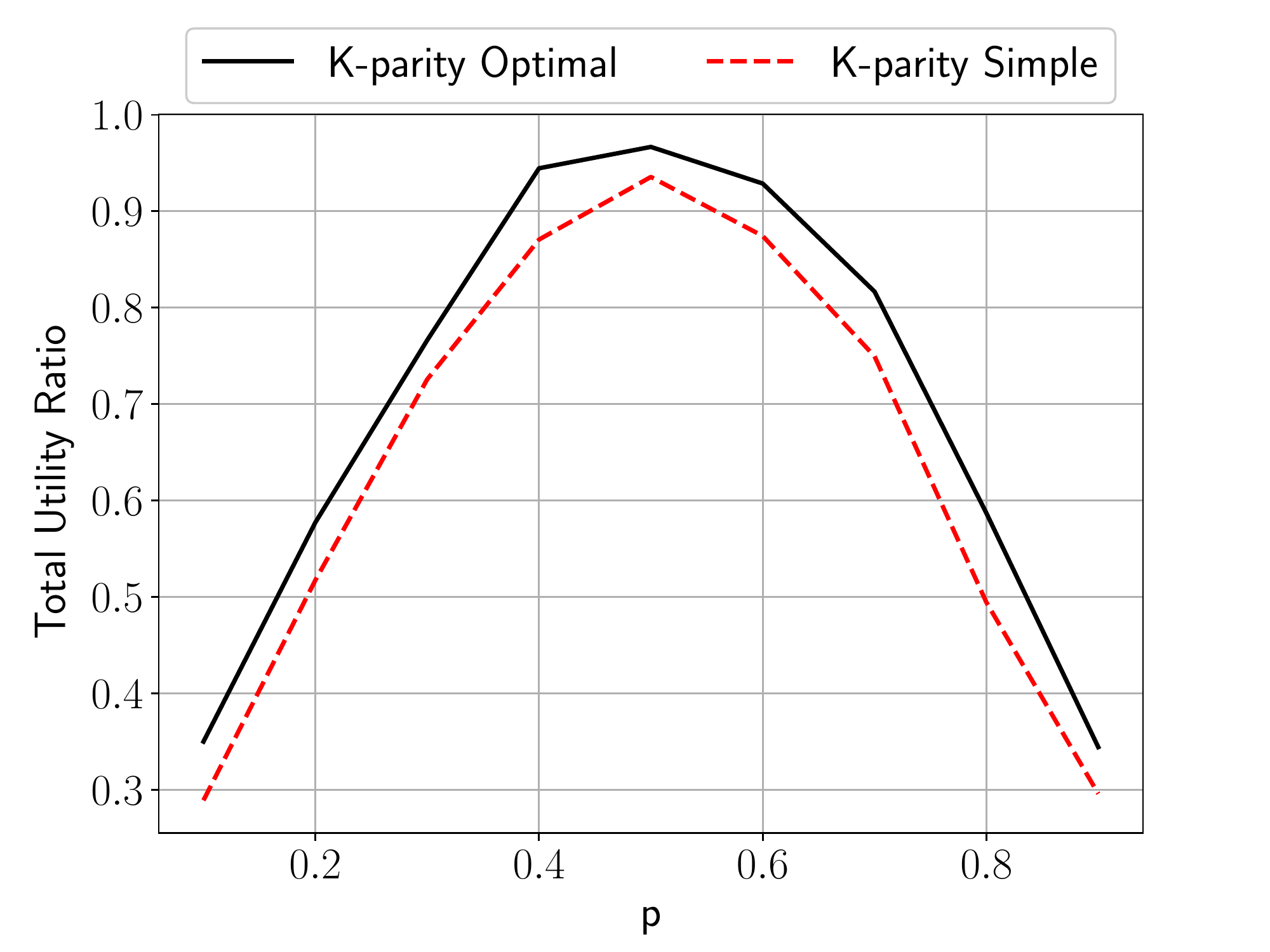}
  \caption{Costs for $10$-parity for the optimal strategy and for simple immediate value}
  \label{fig:kcompare-yahoo}
\end{subfigure}
\hfill
\begin{subfigure}{.45\linewidth}
 \centering
  \includegraphics[scale=.3]{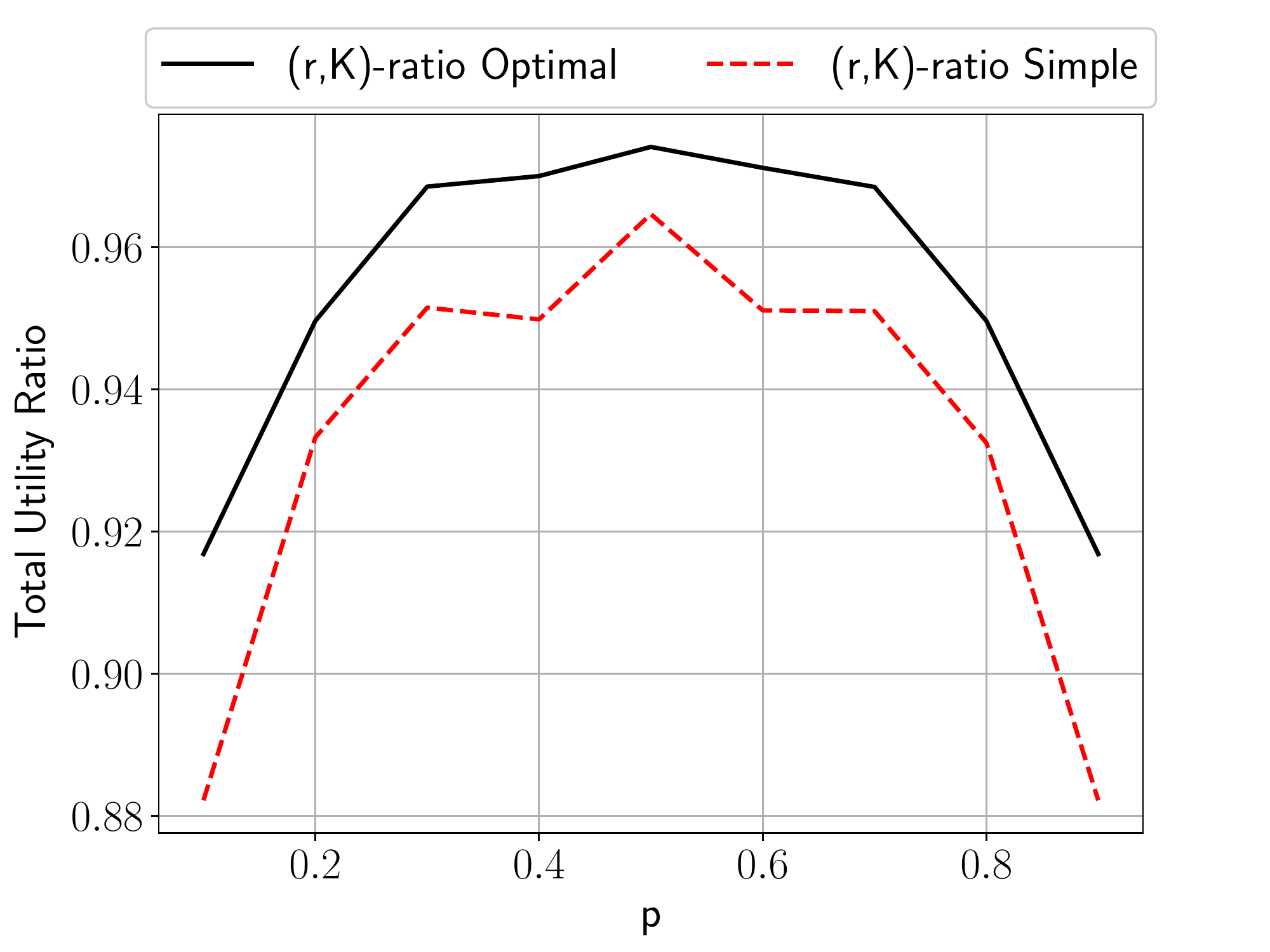}
  \caption{Costs for $(0.8,5)$-ratio for the optimal strategy and for simple immediate value }
  \label{fig:rcompare-yahoo}
\end{subfigure}
\hfill
\caption{Experimental results for Yahoo A1 bidding dataset}
\label{fig:yahoo}
\end{figure*}
 
Each keyword in our dataset has a different bid value distribution and
the restricted adviser can model each keyword separately.
We use the similar approach in our simulations and for each simulation
we compute the optimal bidding strategy for a specific keyword.
We assume that restricted advertiser updates his model parameters
every two days.
 
As mentioned, the Yahoo A1 dataset does not contain the exact
timestamps.
Therefore, we cannot exactly determine which advertisers participated
in any single ad auction.
We randomly select a set of advertisers' bids from each 15 minutes
interval for each of our ad auctions.
Since the dataset does not include information about the viewers, we
sample the bids for both female and male viewers from the same set of
bid values, making their values equal.
 
Figures~\ref{fig:kparity-yahoo} and~\ref{fig:rratio-yahoo} show the
total utility ratio of the $K$-parity and $(r, K)$-ratio versions to
the unrestricted version of the advertiser $i$ for various values of
$K$, $r$, and $p$ on Yahoo A1 bid dataset.
Here, and in the other simulations, we compute this ratio by
simulating restricted and unrestricted
versions of the advertiser $i$, using the same draw of values across the
two versions.
We do this 100 times, computing the average of total utilities $U_i$ for each version.
We then plot the ratio of these two averages.
Since the value of ad slots for both female and male viewers are equal, the total utility of an
unrestricted advertiser will not change for different values of $p$.
On the other hand, a restricted advertiser will get different utilities
based on the distribution of the men and women viewers.
$K$-parity and $(r, K)$-ratio constraints are harder to achieve for extreme values of $p$.
Turning to the effects of $K$,
the results show that when $K$ is large, the $K$-parity advertiser can
reach the utility of the unrestricted advertiser.
Also by relaxing $r$, $r$-ratio advertiser achieves higher utility.
To show the benefit of our approach compared to simply bidding
immediate values, we compare the utility ratio both approaches.
Figure~\ref{fig:kcompare-yahoo} shows that our bidding strategy allows
the advertiser achieve a higher utility.

\subsection{Synthetic Data}
 
A major limitation of the real dataset for our purposes is that it
does not show which ad slots are for men and which for women.
Thus, we use a synthetic dataset to explore how changing their
relative values affects the advertiser's utility.
We generate two synthetic datasets using a log-normal distribution to
sample the advertisers bids.
Table~\ref{tab:values} shows the model parameter settings used for the
two scenarios.
 
\begin{table}
\centering
\caption{Parameters for the log-normal distribution used in the modeling the bids in the ad slot auctions.  $\sigma^2$ is always $0.7$.}
\label{tab:values}
\begin{tabular}{@{}lcccc@{}}
\toprule
Name & \multicolumn{2}{c}{Others}& \multicolumn{2}{c}{Advertiser $i$} \\
& $\mu_{-i}^{\msf{m}}$ & $\mu_{-i}^{\msf{w}}$ & $\mu_i^{\msf{m}}$ & $\mu_i^{\msf{w}}$\\
\midrule
\femalevalue & -2.8 & -2.8  & -3.5  & -2.4\\
\femaleprice & -3.5 & -2.4  & -2.8  & -2.8\\ 
\bottomrule
\end{tabular}
\end{table}
\begin{figure*}[t]
\centering
%
%
\hfill
\begin{subfigure}{.32\linewidth}
\centering
\includegraphics[scale=0.3]{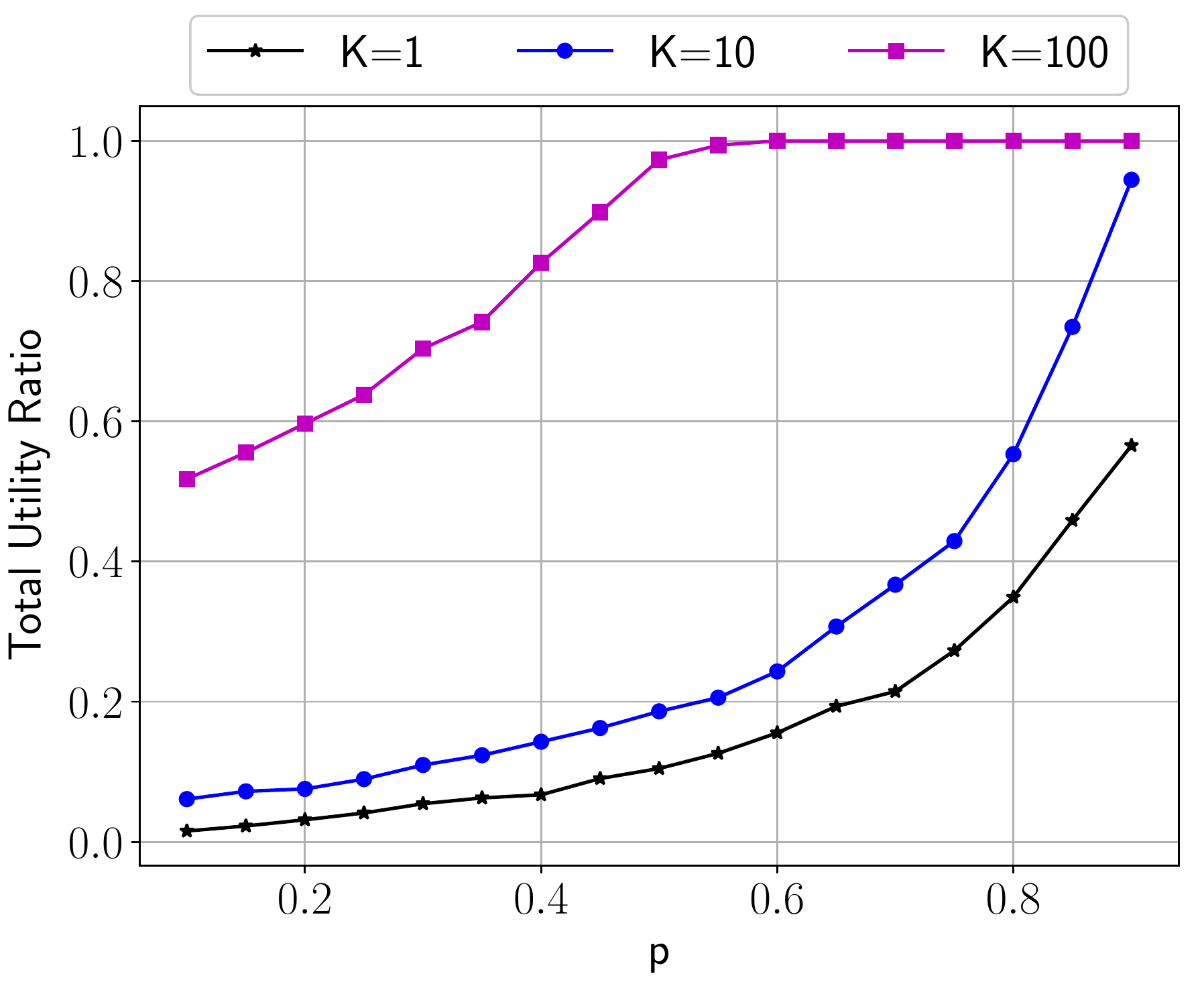}
\caption{Cost of $K$-parity for \emph{\femalevalue} scenario}
\label{fig:diff_p_Kvalues}
\end{subfigure}
\hfill
\begin{subfigure}{.32\linewidth}
\centering
\includegraphics[scale=0.3]{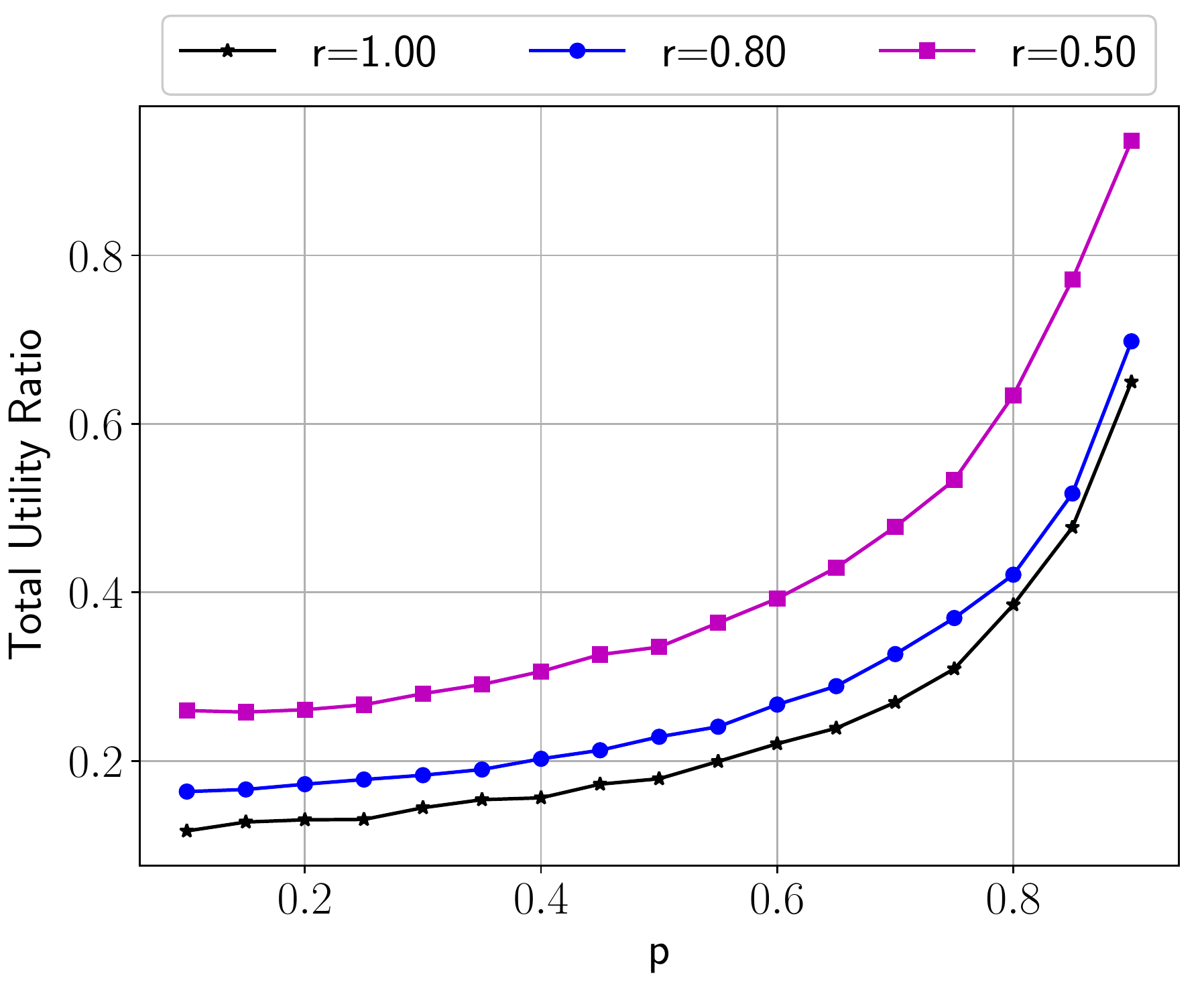}
\caption{Cost of $(r,5)$-ratio for \emph{\femalevalue} scenario}
\label{fig:diff_p_Kvalues}
\end{subfigure}
\hfill
\begin{subfigure}{.32\linewidth}
\centering
\includegraphics[scale=0.3]{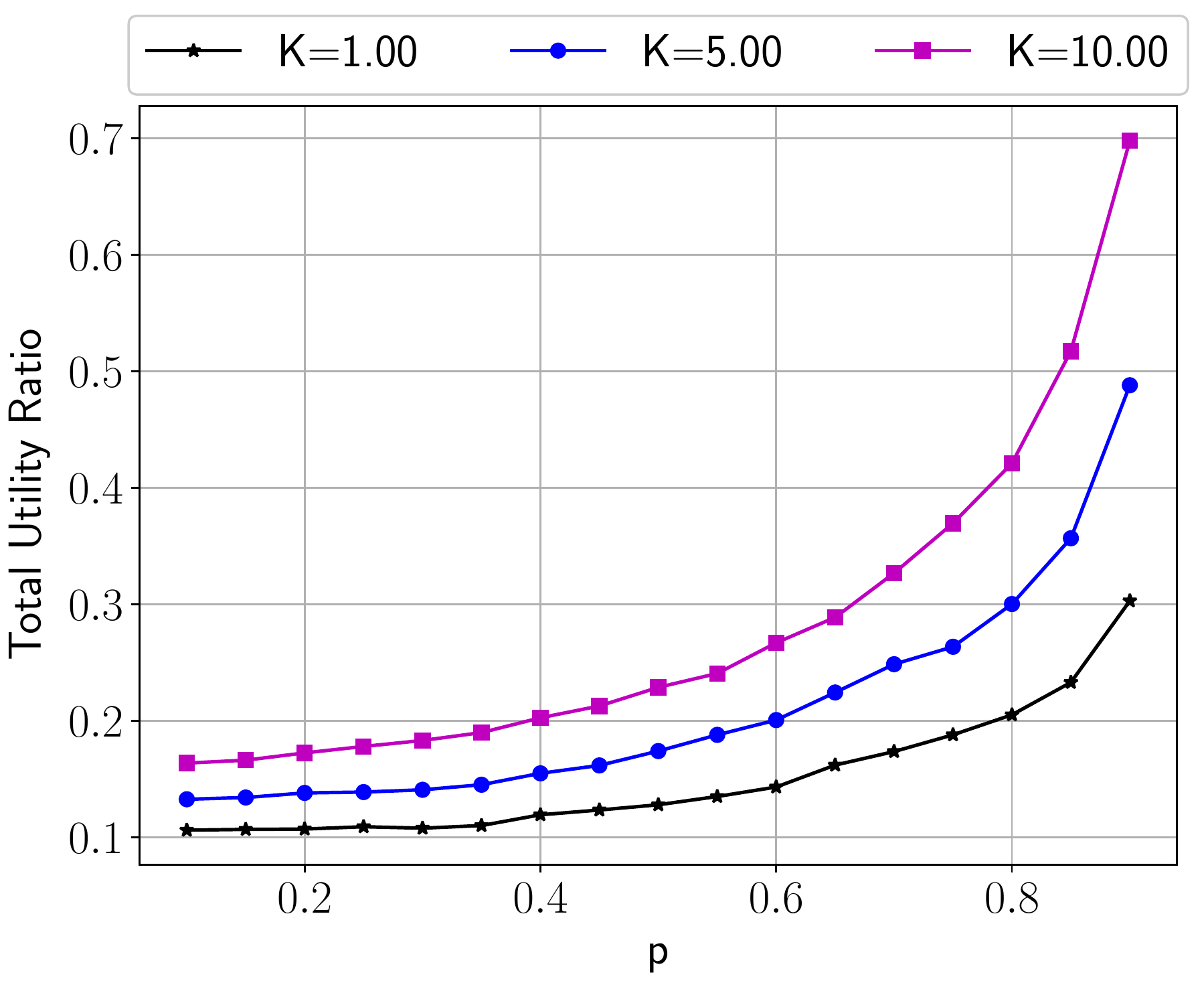}
\caption{Cost of $(0.8,K)$-ratio for \emph{\femalevalue} scenario}
\label{fig:diff_p_Kvalues}
\end{subfigure}
\hfill

\hfill
\begin{subfigure}{.4\linewidth}
\centering
\includegraphics[scale=0.3]{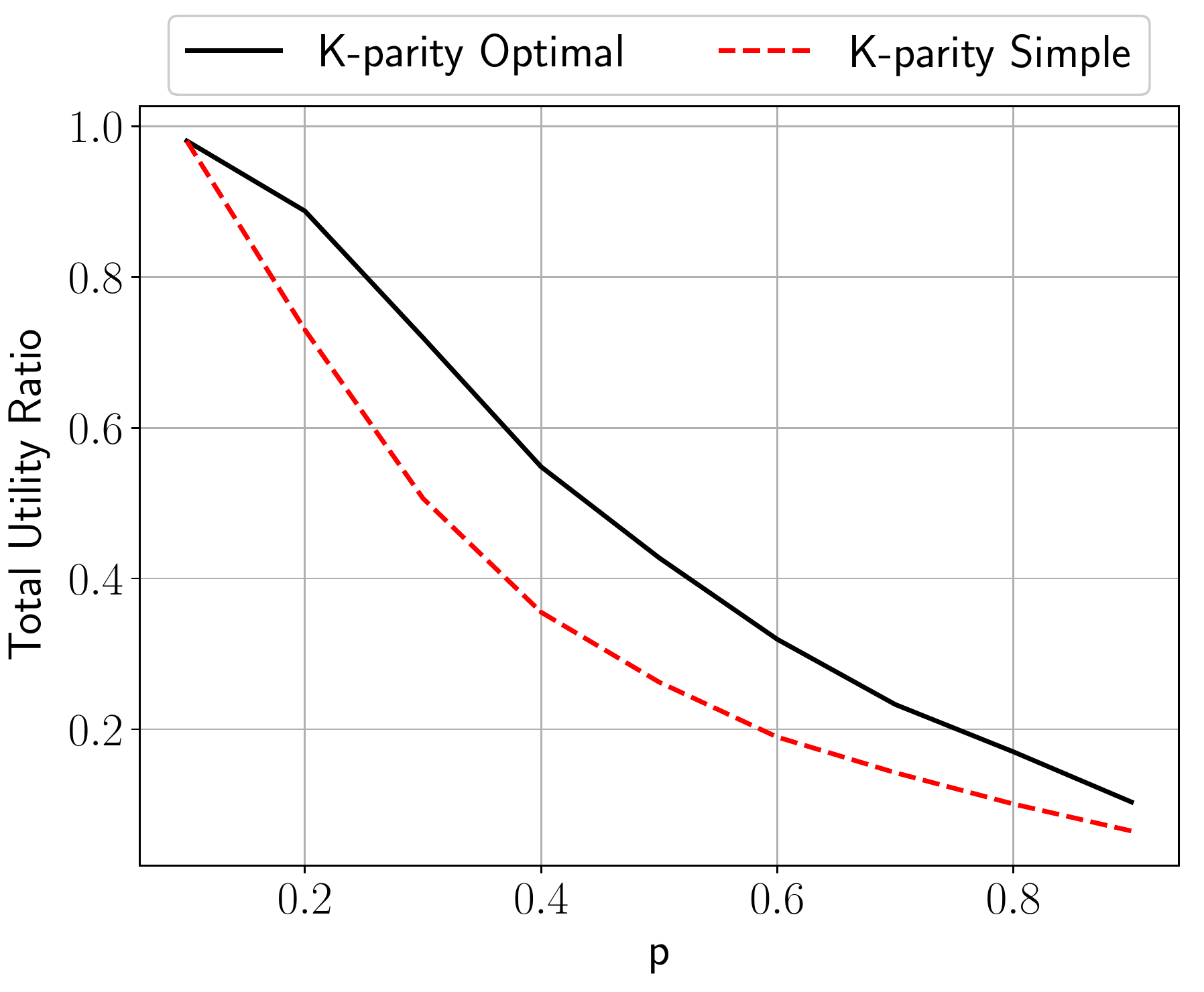}
\caption{Costs for $10$-parity for the optimal strategy and for simple immediate value bidding for \emph{\femaleprice} scenario}
\label{fig:diff_market_values}
\end{subfigure}
\hfill
\begin{subfigure}{.4\linewidth}
\centering
\includegraphics[scale=0.3]{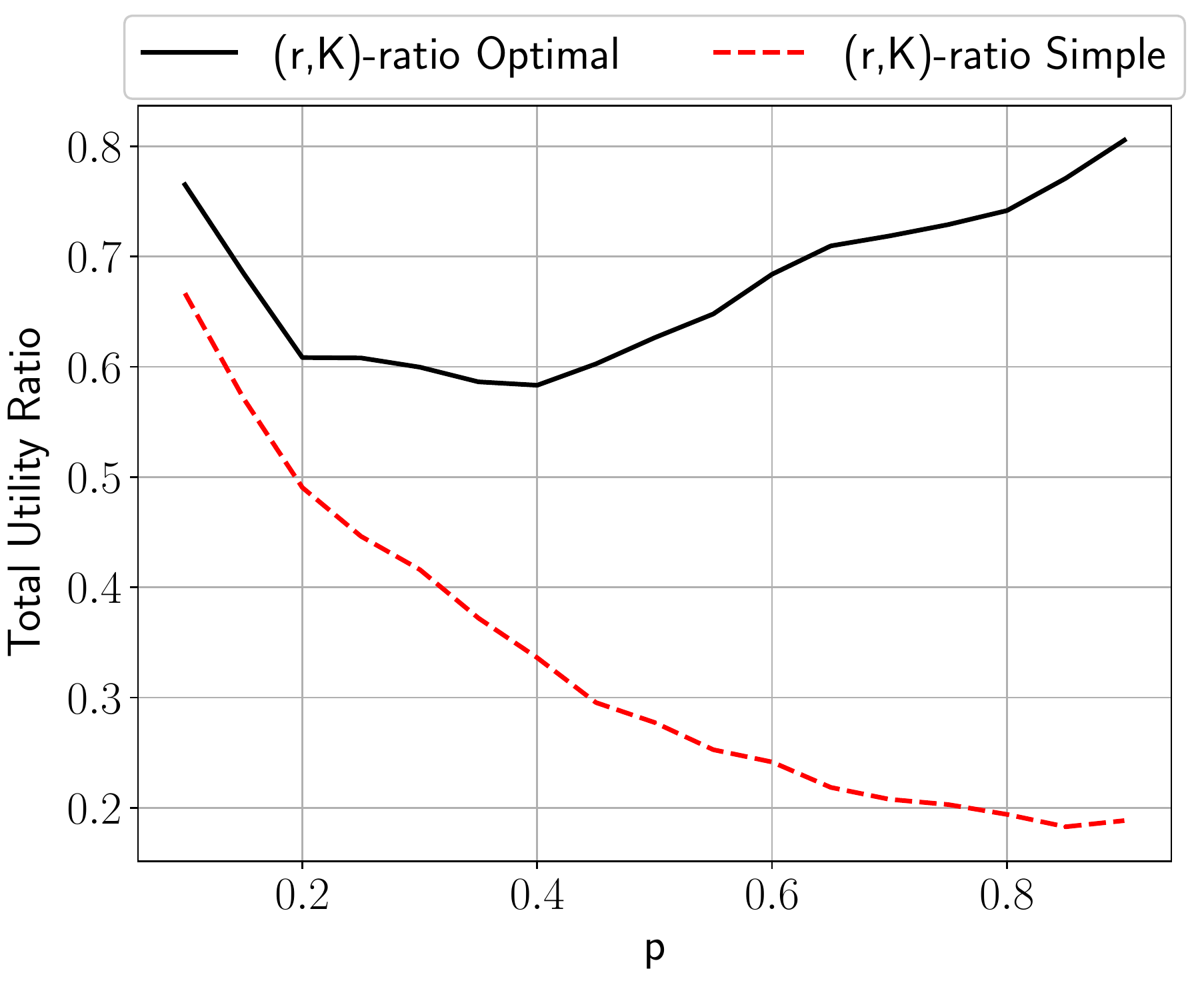}
\caption{Costs for $(0.8,5)$-ratio for the optimal strategy and for simple immediate value bidding for \emph{\femaleprice} scenario}
\label{fig:diff_market_values_ratio}
\end{subfigure}
 
\begin{subfigure}{.32\linewidth}
\centering
\includegraphics[scale=0.3]{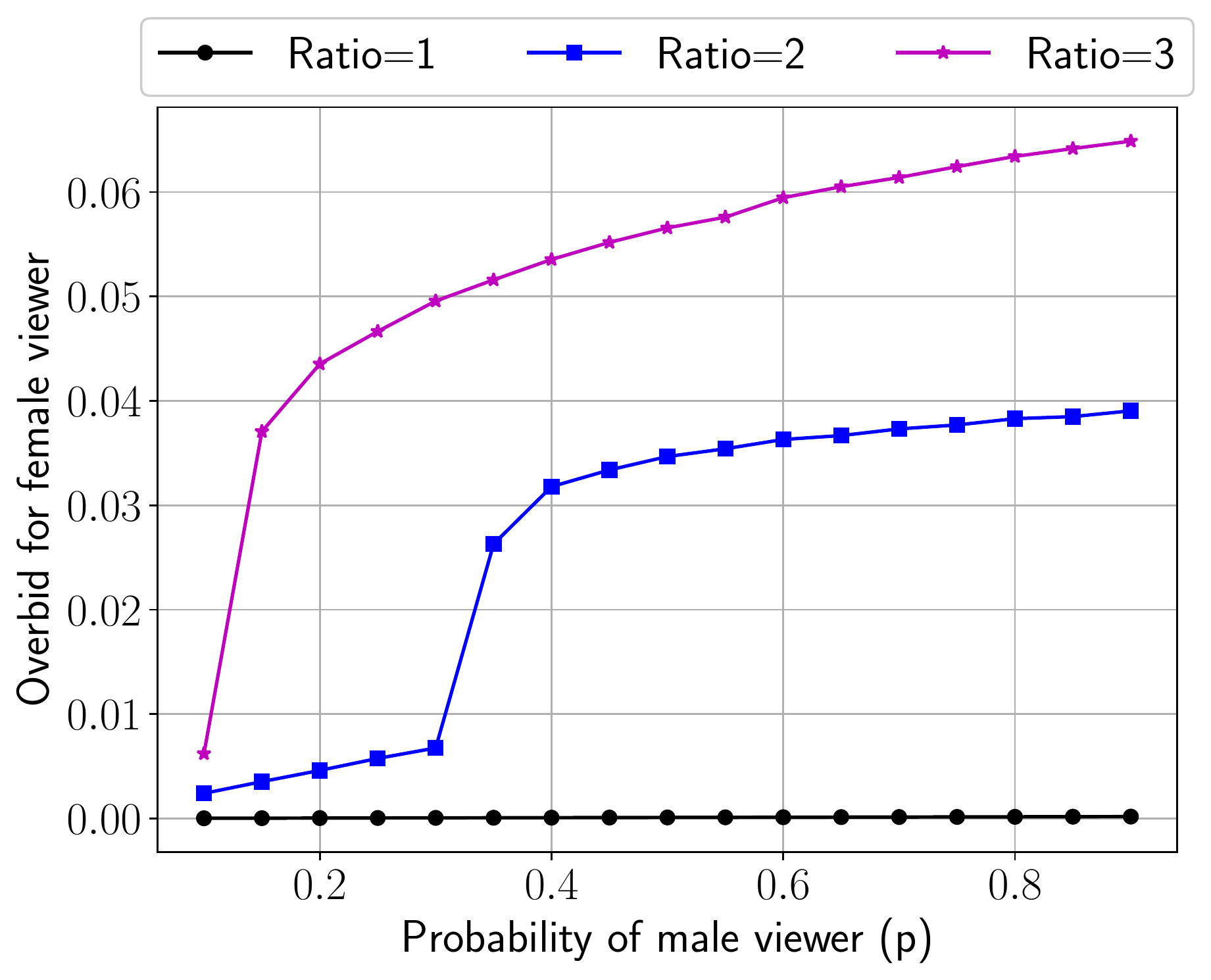}
\caption{Average amount of overbidding for the female viewers for values of $p$}
\label{fig:overbids}
\end{subfigure}
\hfill
\begin{subfigure}{.32\linewidth}
\centering
\includegraphics[scale=0.3]{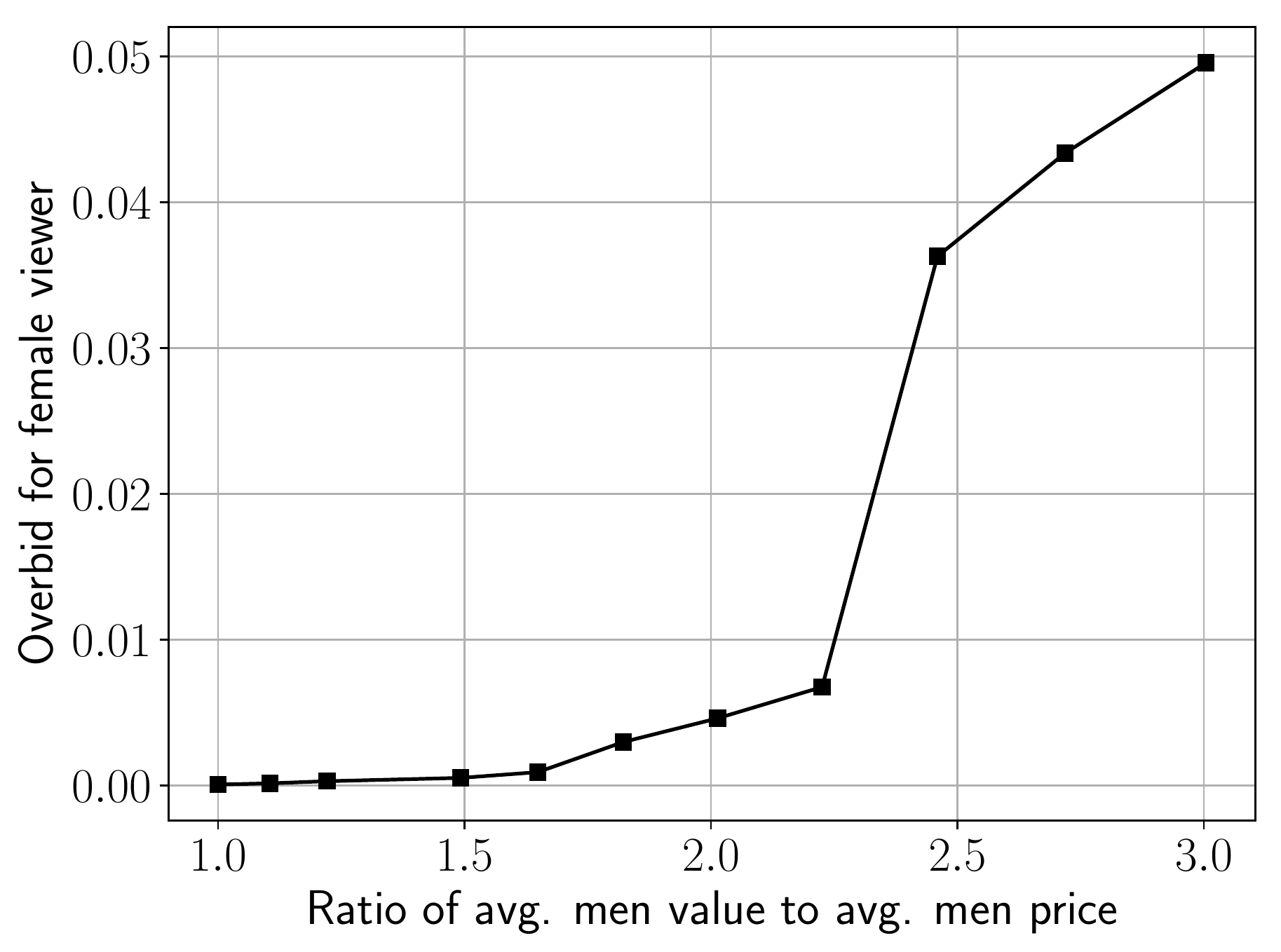}
\caption{Average overbidding for female viewers for ratios of average male viewer value to price of male viewers with $p = 0.5$}
\label{fig:overbids_ratio}
\end{subfigure}
\hfill
\begin{subfigure}{.32\linewidth}
\centering
\includegraphics[scale=0.3]{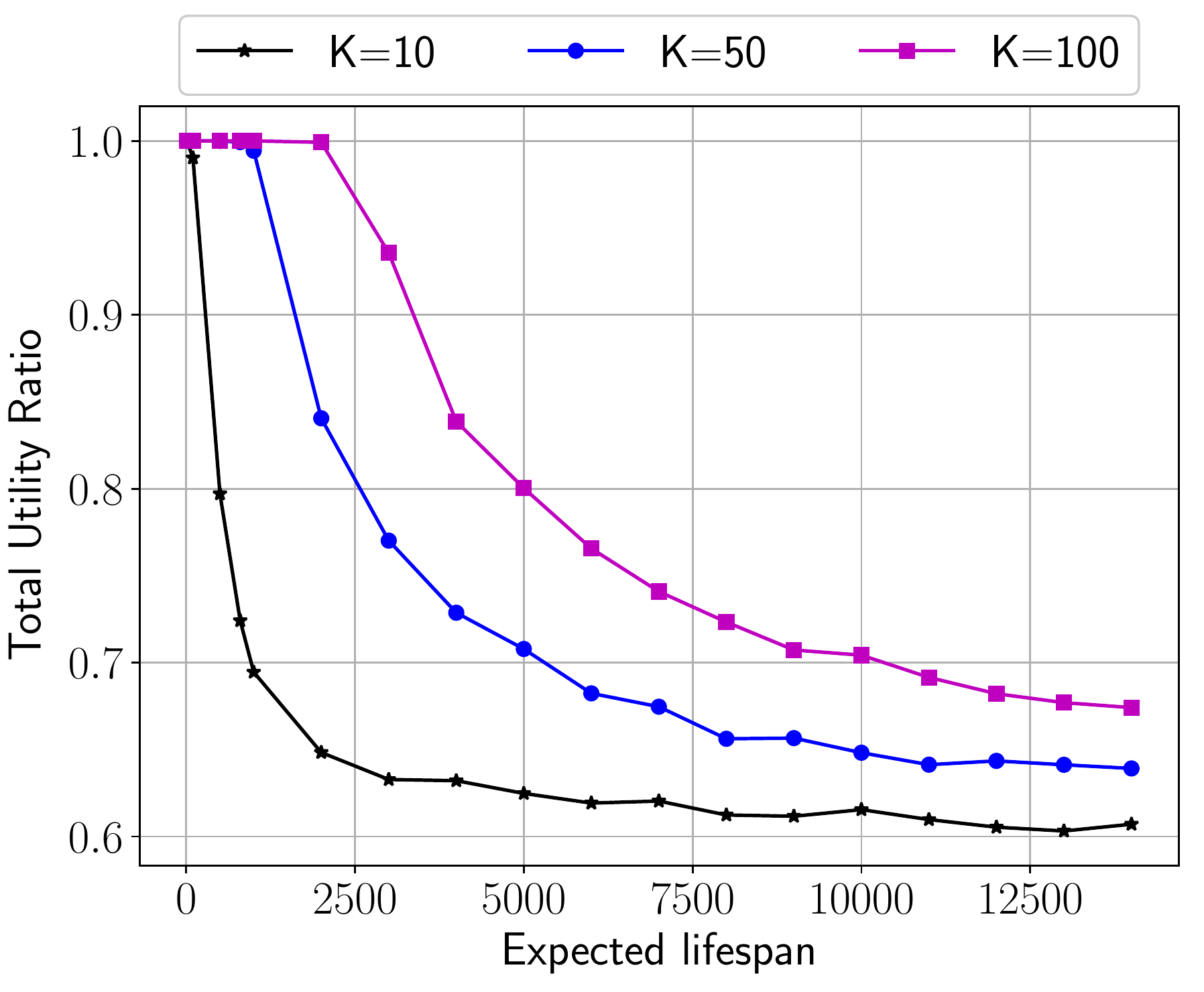}
\caption{Cost of $10$-parity for values of $\delta$ as the expected lifespan of the advertiser in the \emph{\femaleprice} scenario}
\label{fig:lifespan}
\end{subfigure}
\caption{Experimental results for synthetic datasets}
\end{figure*}
To show the effect of assigning different values to men and
women, consider an
advertiser that gives more value to female slots than to male ones,
as shown in the \emph{\femalevalue} parameter settings.
Figure~\ref{fig:diff_p_Kvalues} shows the utility ratio for the
$K$-parity and unrestricted versions of the advertiser in this
scenario.
The $K$-parity version has its maximum utility ratio when
there are more male than female slots.
This may seem counter-intuitive since the advertiser values females
more, but the measured ratio reflects that an
abundance of males means that the $K$-parity version will not have
to operate much differently from the unrestricted one.
This is due to their abundance making overbidding less needed,
decreasing the $K$-parity version's costs.
Lambrecht~et~al.~\citeyearpar{lambrecht18ssrn} empirically showed that
young women are more expensive to show ads to.
To simulate this setting, we considered a scenario in which the other
advertisers prefer females (i.e., $g_j^{\msf{w}} < g_j^{\msf{m}}$ for
all $j \neq i$).
We used the \emph{\femaleprice} parameter settings for this scenario.
As in the first scenario, we have advertiser $i$ value both
types equally, at the average of the two different values used by the
other advertisers.
Figure~\ref{fig:diff_market_values} plots the total utility ratio as
before (solid line).
Note that as women become rare, the $K$-parity version struggles
relative to the unrestricted one since the other advertisers snap up
the few women leaving the constrained version unable to bid for men.
The figure also shows the total utility ratio for a constrained version of
the advertiser $i$ that uses the same simple bidding strategy as the unrestricted
advertisers (dashed line). 
Note that ratio is lower than with our optimal bidding strategy,
showing its value.
This difference comes from our optimal bidding strategy overbidding for the
female viewers, delaying the aforementioned effect. Figure~\ref{fig:diff_market_values_ratio} tells a similar story for
the ratio constraint.

Figures~\ref{fig:overbids} and~\ref{fig:overbids_ratio} further
explore overbidding using a variation on the \emph{\femaleprice}
scenario.
Rather than keep the value that the advertiser $i$ assigns to males
fixed at $\mu_i^{\msf{m}} = -2.5$, we vary this value to see its
effect on overbidding.
Rather than plot $\mu_i^{\msf{m}}$ itself, we plot the ratio of
$\mu_i^{\msf{m}}$ to the value assigned to males by the other
advertisers.
Figure~\ref{fig:overbids} shows this value ratio by using various lines.
For all
such ratios above $1$, as the rate $p$ of male viewers increases, the
optimal $K$-parity advertiser will increase its overbidding on the female
viewers since they are more scared.
Figure~\ref{fig:overbids_ratio} shows that as $\mu_i^{\msf{m}}$
(and, thus, the male value ratio of advertiser $i$ to the other advertisers)
increases, the overbidding for females increases.
Figure~\ref{fig:lifespan} plots the utility ratio as the value of the
rate $\delta$ at which the advertiser $i$ will leave the ad network
changes.
Rather than plot $\delta$ directly, it plots the expected lifespan of
the advertiser computed from $\delta$.
It shows that for short lived advertisers, $K$-parity has no effect
since the advertiser is unlikely to reach $K$ wins for either gender.
However, the constraint rapidly has an effect as the advertiser lives
long enough to win this number of slots.
 
\paragraph{Ad Exchange Revenue}
Also important is how our strategy impacts the revenue of the ad exchange.
We explored the ratio of the ad exchange's revenue when there is one restricted advertiser for each ad slot auction to the case where all advertisers are unrestricted for all of our scenarios ( both real and synthetic dataset). In most cases the ad exchange revenue will not decrease at all. The worst case happens for $(1.0,1)$-ratio constraint advertisers on Yahoo!/@ A1, the ratio of revenues is $0.993$. The ad exchange can have a lower bound on the $K$ and $r$ to make sure it does not lose any  revenue. Therefore, implementing this feature will not significantly reduce the ad exchange's revenue.  Our observations show restricted advertisers are more likely to overbid which increases the ad exchange revenue. Figure~\ref{fig:revenue} we compare the revenue of ad exchange's for different number of restricted advertisers ($\rho$) For Yahoo! A1 dataset. As expected by increasing the number of restricted advertisers the the revenue of ad exchange's increases.

\begin{figure}
 \includegraphics[scale=0.35]{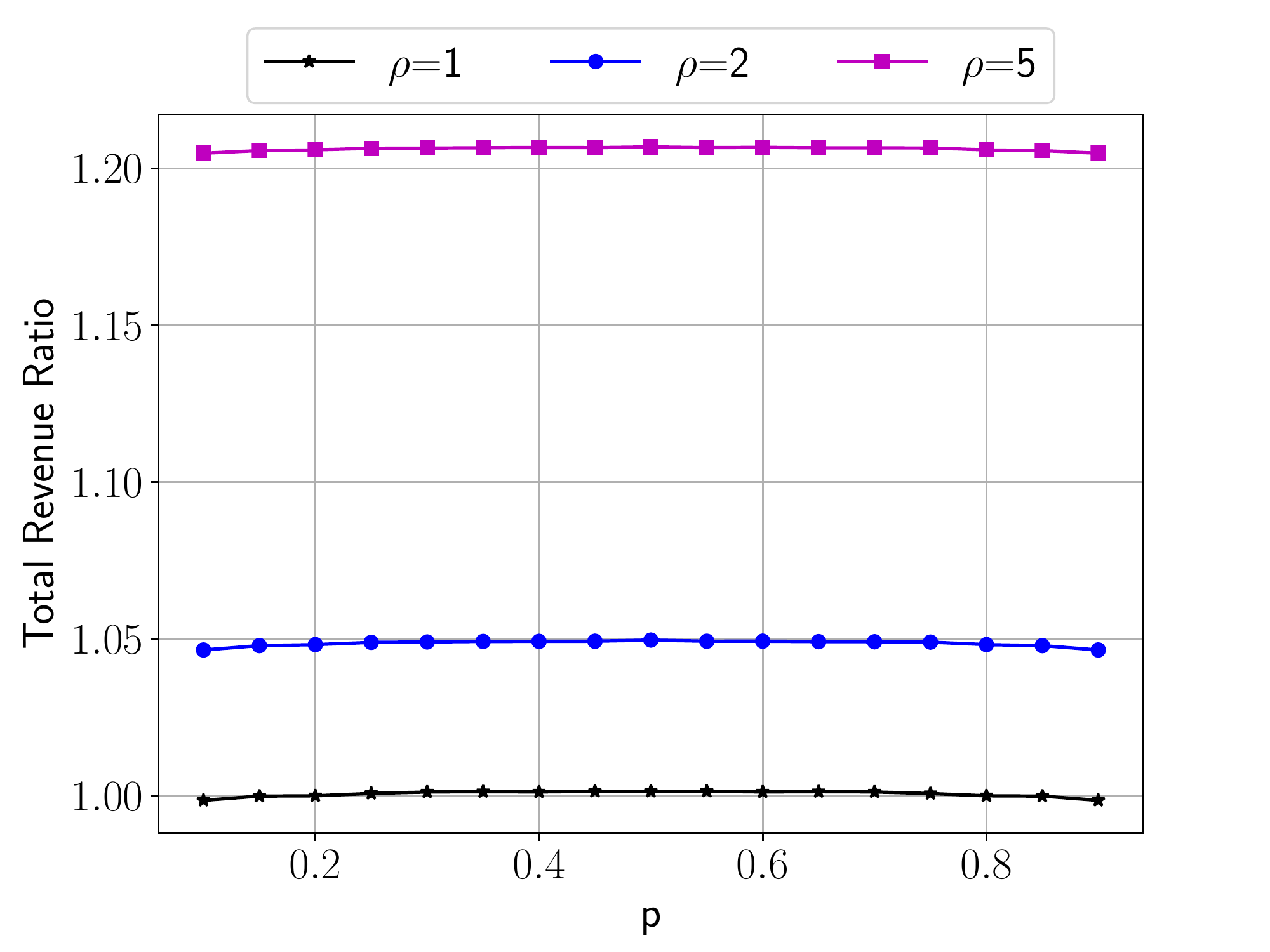}
 \caption{Ratio of the ad plot form revenue for Yahoo A1 dataset with (0.8,5)-ratio restricted advertiser for different number of restricted advertisers.}
 \label{fig:revenue}
\end{figure}
 
\paragraph{Performance}
Algorithms~\ref{alg:converge} and~\ref{alg:converge_ratio}, each of
which only has to run once for each parameter setting, completed in
under 2~minutes and under 10~minutes, respectively.
Calculating bids during auctions, each took the 2~microseconds.
We used a 2013 MacBook Pro with
a 2.3\,GHz Intel Core~i7 and 16\,GB of 1600\,MHz DDR3 memory.

\section{Conclusion and Discussion}
\label{sec:discussion}

Adding parity constraints results in a surprisingly complex bidding
problem, exhibiting both over- and underbidding relative to the
advertiser's immediate value of an ad slot.
Despite this complexity, we show a practical way of computing optimal
bids, to within a small approximation factor $\epsilon$.
This enables us to characterize how the cost of parity depends upon
not just its level of strictness $K$ or $R$, but also the base rate $p$ of
types, their relative values to both the governed advertiser $i$ and
to other advertisers, and the lifespan (or discounting factor)
$\delta$, in sometimes counter-intuitive ways.

We envision two ways in which advertisers could use our bidding strategy.
Firstly, ad exchanges might implement it for them as a
feature in the ad buying interface.
Such exchanges could use the data it has to determine the demographics
of individuals viewing ad slots and adjust bids accordingly.
While this would require a change to the ad exchange, it would not
require modifying the core auction mechanism, making it a more
straightforward feature to add.

Secondly, the strategy could be used either directly by the advertiser
or offered to them by demand-side platforms as a feature.
This approach has the advantage of not requiring any changes to the ad
exchange.
It has the disadvantage of only working for ad exchanges that support
real-time bidding and programmatic advertising with rich enough data
to infer the group membership of the people viewing ad slots.
Additionally, such rich data can pose privacy concerns.

We believe that either of these approaches to deployment would be more
straightforward than any way of deploying an auction mechanism that
enforces parity constraints~\cite{celis19arxiv} or
Guaranteed ad Delivery
(GD)~\cite{salomatin2012unified,turner2012planning}.
Only the ad exchange would be able to implement such functionality.
Presumably, ad exchanges have already selected the auction mechanism
that they believe would be best for their business and would be
reluctant to change it in a way that could have wide ranging effects.
Given that Google uses a generalization of second-price
auctions~\cite{google18adsense}, it may believe that the theoretical
result that second-price auctions are uniquely optimal in
certain settings has some bearing on its setting.
Thus, it may believe that any change to its auction mechanism is
likely to reduce its profits, a strong disincentive.
We believe that ad exchanges would be more willing to implement a
change that instead only alters the bids of advertisers who opt in
since it would be equivalent to one that advertisers could already
implement unilaterally by altering their bids.
Furthermore, since our approach changes just opted-in
advertisers' bids, there is a sense in which they pay for it.

Future work can explore more complex forms of nondiscrimination
constraints, such as ones holding probabilistically or
asymptotically.
The use of bonuses for ad clicks and online tracking to assign
different expected values to individual ad slots could be
considered.
Future work could accommodate constraints for non-binary sensitive
attributes, such as location (a proxy for race, which is
apparently not explicitly tracked by any ad exchange)
or for multiple constraints simultaneously.
Although our MDPs can straightforwardly be extended to such cases
using a cross-product-like construction, the MDP size will be
exponential in the number of constraints and their values,
motivating more significant future work.

The constraints we explore are very strict in that
they must hold at all times, as opposed to holding with high
probability or asymptotically, which might be acceptable in some settings.
In related problems, parity may only be required at the end of certain
checkpoints, such as at the end of a hiring season.
Exploring such relaxations can be future work.

We used a simple model in which the expected value of each female
slot is equal to the others, and the expected value of each male slot
is equal to the others.
Advertisers can use online tracking, machine learning, and other
techniques to compute more fine-grained estimations of slot
values.
Furthermore, our model of ad exchanges does not include that they are
often paid more when the viewer clicks on the ad.
Thus, their expected value for selling an slot to an advertiser
depends upon not just the bid prices but also the fits of the ads for
the slot, which also can be estimated with online tracking and
machine learning.
Such tracking and machine learning can be another route to
discrimination~\cite{datta18fatstar}.



\paragraph*{Acknowledgements}
Milad Nasr is supported by a Google PhD Fellowship
in Security and Privacy. We gratefully acknowledge funding support from
the National Science Foundation
(Grant 1237265).
The opinions in this paper are those of the authors and do not
necessarily reflect the opinions of any funding sponsor or the United
States Government.

\bibliographystyle{named}
\bibliography{bibliography}

\end{document}